\documentclass[twocolumn,10pt,superscriptaddress,amsmath,amssymb,aps,pra,floatfix,]{revtex4-2}

\usepackage{graphicx}
\usepackage{bm}
\usepackage{times}
\usepackage{amsmath,amssymb}
\usepackage{mathtools}
\usepackage{amsthm}
\usepackage{enumerate}
\usepackage{multirow}
\usepackage{siunitx}
\usepackage{makecell}
\usepackage{physics}
\usepackage{longtable,array}
\usepackage[caption=false]{subfig}
\usepackage{tikz}
\usepackage{xcolor}
\usepackage{tabularx}
\usepackage{soul}
\usepackage{booktabs}
\setstcolor{red}

\renewcommand{\arraystretch}{1.2}

\usepackage[colorlinks=true,linkcolor=blue,citecolor=red, linktocpage=true,breaklinks=true]{hyperref}

\newcommand{\eq}{\begin{equation}}
\newcommand{\en}{\end{equation}}
\newcommand{\eqa}{\begin{eqnarray}}
\newcommand{\ena}{\end{eqnarray}}
\newcommand{\pr}{\operatorname{Pr}}

\newtheorem{theorem}{Theorem} 
\newtheorem{lemma}[theorem]{Lemma} 
\newtheorem{corollary}[theorem]{Corollary}

\newtheorem{proposition}[theorem]{Proposition}

\begin{document}

\title{Circuit-depth optimization of quantum partial-search algorithms}
\author{Yan-Bo \surname{Jiang}}
\affiliation{School of Physics, Northwest University, Xi'an 710127, China}

\author{Xiao-Hui \surname{Wang}}
\affiliation{School of Physics, Northwest University, Xi'an 710127, China}
\affiliation{Shaanxi Key Laboratory for Theoretical Physics Frontiers, Xi'an 710127, China}
\affiliation{Peng Huanwu Center for Fundamental Theory, Xi'an 710127, China}
\affiliation{Fundamental Discipline Research Center for Quantum Science and Technology of Shaanxi Province, Xi'an 710127, China}

\author{Kun \surname{Zhang}}
\email{kunzhang@nwu.edu.cn}
\affiliation{School of Physics, Northwest University, Xi'an 710127, China}
\affiliation{Shaanxi Key Laboratory for Theoretical Physics Frontiers, Xi'an 710127, China}
\affiliation{Peng Huanwu Center for Fundamental Theory, Xi'an 710127, China}
\affiliation{Fundamental Discipline Research Center for Quantum Science and Technology of Shaanxi Province, Xi'an 710127, China}

\author{Vladimir \surname{Korepin}}
\affiliation{C. N. Yang Institute for Theoretical Physics, Stony Brook University, New York 11794, USA}

\date{\today}

\begin{abstract}
Grover's algorithm is optimal in terms of oracle queries. We can trade accuracy for speed, which gives rise to the quantum partial-search algorithm. The partial-search algorithm is implemented using two kinds of Grover operators, global and local, where the former is the standard Grover operator for full search and the latter has a diffusion operator acting on the search subspace. The global-local-global sequence, also known as the Grover-Radhakrishnan-Korepin (GRK) algorithm, has been proved optimal in the oracle-query metric. In this work, we show that the alternating sequence of global and local Grover operators, formed by repeatedly applying a fixed product of global and local Grover operators, can achieve a lower expected circuit depth than the GRK algorithm. Through systematic analysis, we obtain its exact success probability, expected depth, and asymptotically optimal parameters. We derive the boundary, characterized by the ratio between the depths of the oracle and the global diffusion operator, separating the depth-optimal and oracle-optimal partial-search algorithms. When the depths of the oracle and the global diffusion operator are comparable, our proposed alternating partial-search sequence can reduce the minimum expected circuit depth by more than 20\% compared with the GRK algorithm.
\end{abstract}

\maketitle

\section{\label{sec:intro} Introduction}
Unstructured search is a canonical example of a provable quantum advantage. Grover's algorithm searches a database of size $N$ using $\Theta(\sqrt N)$ oracle queries, a quadratic improvement over the classical $\Theta(N)$ scaling, and this query complexity is optimal \cite{grover1996fast,grover1997quantum,bennett1997strengths,zalka1999grover}. Its underlying amplitude-amplification mechanism \cite{brassard2000quantum} also supports applications in global optimization \cite{baritompa2005grover}, cryptanalysis \cite{grassl2016applying}, Boolean satisfiability \cite{dantsin2005quantum}, and quantum machine learning \cite{kerenidis2019qmeans}. Experimental realizations began with a two-qubit nuclear-magnetic-resonance implementation \cite{chuang1998experimental} and have since progressed across trapped ions \cite{figgatt2017complete}, superconducting circuits \cite{dewes2012quantum,zhang2021implementation,yan2023scalable,abughanem2025characterizing}, silicon spin qubits \cite{thorvaldson2025grover}, and error-suppressed Grover search on up to five qubits \cite{pokharel2024better}. Despite this progress, implementations of Grover search remain small-scale demonstrations, reflecting the gap between an asymptotic query advantage and the resource limitations of noisy intermediate-scale quantum devices \cite{preskill2018quantum,leymann2020bitter}.

When only partial information is required, quantum partial search seeks to identify the block containing the target rather than the target itself \cite{grover2005partial}. Partial search can be applied to partitioned databases, prefix-based recovery, and hierarchical search \cite{korepin2007hierarchical,korepin2009quantum}. For a database of size $N=2^n$ partitioned into $K=2^{n-m}$ blocks of size $b=2^m$, the Grover-Radhakrishnan-Korepin (GRK) algorithm combines global Grover iterations with local iterations, where the latter apply the diffusion operator only to the $m$ within-block qubits \cite{grover2005partial,korepin2005simple,korepin2005optimization,korepin2006quest}. The resulting GRK algorithm has a global-local-global form whose asymptotic optimality in the oracle-query metric has recently been established \cite{jiang2026exactboundsquantumpartial,zhang2026asymptotic}. 

Although both Grover's full-search algorithm and the GRK partial-search algorithm have been proven optimal in terms of the number of oracle queries, query optimality does not imply depth optimality. This is because the oracle is not the only operator used in implementations. The depth-efficient design of the diffusion operator for full search has been systematically investigated \cite{Zhang2020DepthO,brianski2021introducing,Liu2021HardwareEfficientQuantumSearch,zhang2026palindromic,burke2026quantum}. However, the depth-efficient design of partial-search algorithms has not been explored. Intuitively, we can increase the number of local diffusion operators to reduce the minimum expected circuit depth. Since the global and local Grover operators do not commute, their order is essential for efficiency.

In this work, we address the depth-optimal design of partial-search algorithms. The key metric is the expected depth required to find the target block rather than the number of oracle queries. Through finite-size exhaustive numerical tests over all possible sequences of global and local Grover operators, we identify depth-optimal sequences in which the global and local Grover operators are applied alternately. We systematically analyze the alternating sequence and derive its exact success probability, expected depth, and asymptotically optimal parameters. We then compare the depth-optimized GRK family with the proposed depth-efficient alternating sequence. When the depths of the oracle and the global diffusion operator are comparable, we find an asymptotic 25\% reduction in the minimum expected circuit depth compared with the oracle-optimal GRK algorithm.


The paper is organized as follows. Sec.~\ref{sec:quantum_search} reviews Grover search and the GRK partial-search algorithm. Sec.~\ref{sec:numerical} describes the numerical optimization and identifies the structure of the depth-efficient partial-search operator. Sec.~\ref{sec:analytical} derives the exact and asymptotic theory of the partial-search algorithm with the alternating sequence. Sec.~\ref{sec:critical} compares the proposed alternating sequence with the GRK algorithm and determines the critical depth ratio, which separates the depth- and oracle-optimal regimes of the partial-search problem. The final section presents the conclusions and outlook. Supporting results are presented in the appendices.

\section{\label{sec:quantum_search} Quantum search algorithms}

In this section, we review the standard Grover search algorithm and the GRK partial-search algorithm in Secs.~\ref{subsec:grover} and~\ref{subsec:partial_search}, respectively.

\subsection{Grover's algorithm\label{subsec:grover}}

Let the database size be $N=2^n$. Let $t\in\{0,1\}^n$ denote the target specified by a Boolean function $f(w)$ with $f(t)=1$ and $f(w)=0$ for $w\neq t$. Grover search starts from the uniform state
\begin{equation}
\label{eq:s_n}
|s_n\rangle=\frac{1}{\sqrt{N}}\sum_{w=0}^{N-1}|w\rangle.
\end{equation}
The phase oracle acts as \cite{nielsen2010quantum}
\begin{equation}
O_t|w\rangle=(-1)^{f(w)}|w\rangle,
\qquad O_t=I_n-2|t\rangle\langle t|,
\end{equation}
where $I_n$ is the identity operator on $n$ qubits. The phase oracle alone does not change measurement probabilities. Amplitude amplification is produced by the diffusion operator
\begin{equation}
D_n=2|s_n\rangle\langle s_n|-I_n,
\end{equation}
which reflects amplitudes about the uniform state. Note that both the oracle $O_t$ and the diffusion operator $D_n$ are highly nonlocal. The circuit depth required to implement $O_t$ is problem-dependent. The circuit depth of $D_n$ is determined mainly by the implementation of the multi-controlled Toffoli gate, for which ancilla-free constructions of depth linear in $n$ are available \cite{barenco1995,saeedi2013linear,he2017decompositions,zindorf2025efficient}. 

The Grover operator is \cite{grover1996fast,grover1997quantum}
\begin{equation}
G_n=D_nO_t.
\end{equation}
With $\sin\theta_1=1/\sqrt{N}$, each iteration of $G_n$ realizes a rotation through an angle $2\theta_1$ within the two-dimensional subspace spanned by the target state and the uniform non-target state. If $\tilde{\omega}$ denotes the measured $n$-bit string, then after $k$ iterations the probability of measuring the target is
\begin{equation}
\pr_{\tilde{\omega}=t}(k)=\sin^2((2k+1)\theta_1).
\end{equation}
For $N\gg1$, the integer iteration count closest to the first maximum is
\begin{equation}
k_{\mathrm{opt}}\simeq\left\lfloor\frac{\pi}{4}\sqrt N-\frac12\right\rceil.
\end{equation}
Here $\lfloor \cdot\rceil$ denotes nearest-integer rounding. Grover's algorithm gives a quadratic query speedup over classical unstructured search.

It is not optimal to run Grover search to its first probability maximum, because the expected number of oracle queries can be further optimized. Specifically, one may stop earlier and restart after failure, as in serial or punctuated search \cite{boyer1996tight,Gingrich2000Generalized}. For $k\in\mathbb Z_{\ge1}$, the success probability after $k$ iterations is $\pr_{\tilde{\omega}=t}(k)$, which gives the expected iteration count
\begin{equation}
\mathrm{E}_{\tilde{\omega}=t}(k)=\frac{k}{\pr_{\tilde{\omega}=t}(k)}.
\end{equation}
Minimizing this expectation yields an optimal stopping point $k_{\min}$ that satisfies the transcendental equation $\tan((2k_{\min}+1)\theta_1)=4\theta_1 k_{\min}$. In the large-$N$ limit this gives $k_{\min}\simeq (y_0/2)\sqrt N$, where $y_0\approx1.1656$ is the smallest positive root of $\tan y=2y$. The corresponding minimum expected iteration count is
\begin{equation}
\label{eq:E_min_Grover}
\mathrm{E}_{\tilde{\omega}=t}^{\min}=\frac{k_{\min}}{\pr_{\tilde{\omega}=t}(k_{\min})}\simeq\kappa_0\sqrt{N},
\end{equation}
with $\kappa_0=y_0/(2\sin^2y_0)\approx0.6900$, which is smaller than $\pi/4$.

\subsection{\label{subsec:partial_search} Partial-search algorithm}

In partial search, the goal is to identify the block containing the target rather than the target item itself. Consequently, we can trade resolution for speed. We write the target string as $t=t_1t_2$, where $t_1$ has length $n-m$ and labels the block, while $t_2$ has length $m$ and labels the item inside that block. Thus the database is partitioned into $K=2^{n-m}$ blocks, each of size $b=2^m$. The GRK algorithm introduces a local diffusion operator that performs reflections only within individual blocks \cite{grover2005partial,korepin2005simple,korepin2005optimization}:
\begin{equation}
D_m=I_{n-m}\otimes\left(2|s_m\rangle\langle s_m|-I_m\right),
\end{equation}
where $|s_m\rangle$ is the uniform superposition of the $m$-qubit computational-basis states and $I_{n-m}$ and $I_m$ are the identity operators on the block-label and block-item qubits, respectively. Unlike $D_n$, which acts on all $n$ qubits, $D_m$ leaves the first $n-m$ qubits unchanged and diffuses only the remaining $m$ qubits. For $m<n$, local diffusion can therefore reduce the non-oracle contribution to circuit depth.

The local Grover operator is defined as
\begin{equation}
G_m=D_mO_t.
\end{equation}
The local Grover operator enlarges the invariant search subspace from two to three dimensions \cite{korepin2006group}. A convenient orthonormal basis is
\begin{subequations}
\label{eq:three_dimen_basis}
\begin{align}
&|t\rangle=|t_1\rangle\otimes|t_2\rangle,\\
&|b\bar t\rangle=\frac{1}{\sqrt{b-1}}\sum_{j\neq t_2}|t_1\rangle\otimes|j\rangle,\\
&|\bar b\rangle=\frac{1}{\sqrt{N-b}}\left(\sqrt{N}|s_n\rangle-|t\rangle-\sqrt{b-1}|b\bar t\rangle\right),
\end{align}
\end{subequations}
where $|b\bar t\rangle$ is the uniform non-target state inside the target block, and $|\bar b\rangle$ is the uniform state over all non-target blocks. The initial state in Eq. \eqref{eq:s_n} is therefore
\begin{equation}\label{eq:sn_three}
|s_n\rangle=\sin\gamma\sin\theta_2|t\rangle+\sin\gamma\cos\theta_2|b\bar t\rangle+\cos\gamma|\bar b\rangle,
\end{equation}
with the angles $\sin\theta_2 = 1/\sqrt{b}$ and $\sin\gamma = 1/\sqrt{K}$. The matrix representations of the global and local Grover operators in this three-dimensional basis are, respectively,
\begin{widetext}
\begin{equation}\label{eq:matrix_GnGm}
G_n=\begin{pmatrix}
1-2\sin^2\gamma\sin^2\theta_2 & 2\sin^2\gamma\sin\theta_2\cos\theta_2 & 2\sin\gamma\cos\gamma\sin\theta_2\\
-2\sin^2\gamma\sin\theta_2\cos\theta_2 & 2\sin^2\gamma\cos^2\theta_2-1 & 2\sin\gamma\cos\gamma\cos\theta_2\\
-2\sin\gamma\cos\gamma\sin\theta_2 & 2\sin\gamma\cos\gamma\cos\theta_2 & 2\cos^2\gamma-1
\end{pmatrix},\quad
G_m=\begin{pmatrix}
\cos(2\theta_2) & \sin(2\theta_2) & 0\\
-\sin(2\theta_2) & \cos(2\theta_2) & 0\\
0 & 0 & 1
\end{pmatrix}.
\end{equation}
\end{widetext}
These matrices form an $O(3)$ representation \cite{korepin2006group}, with $\det G_n=-1$. Thus $G_n$ is an improper orthogonal transformation rather than an element of $SO(3)$.

\begin{figure}[t]
\centering
\includegraphics[width=\linewidth]{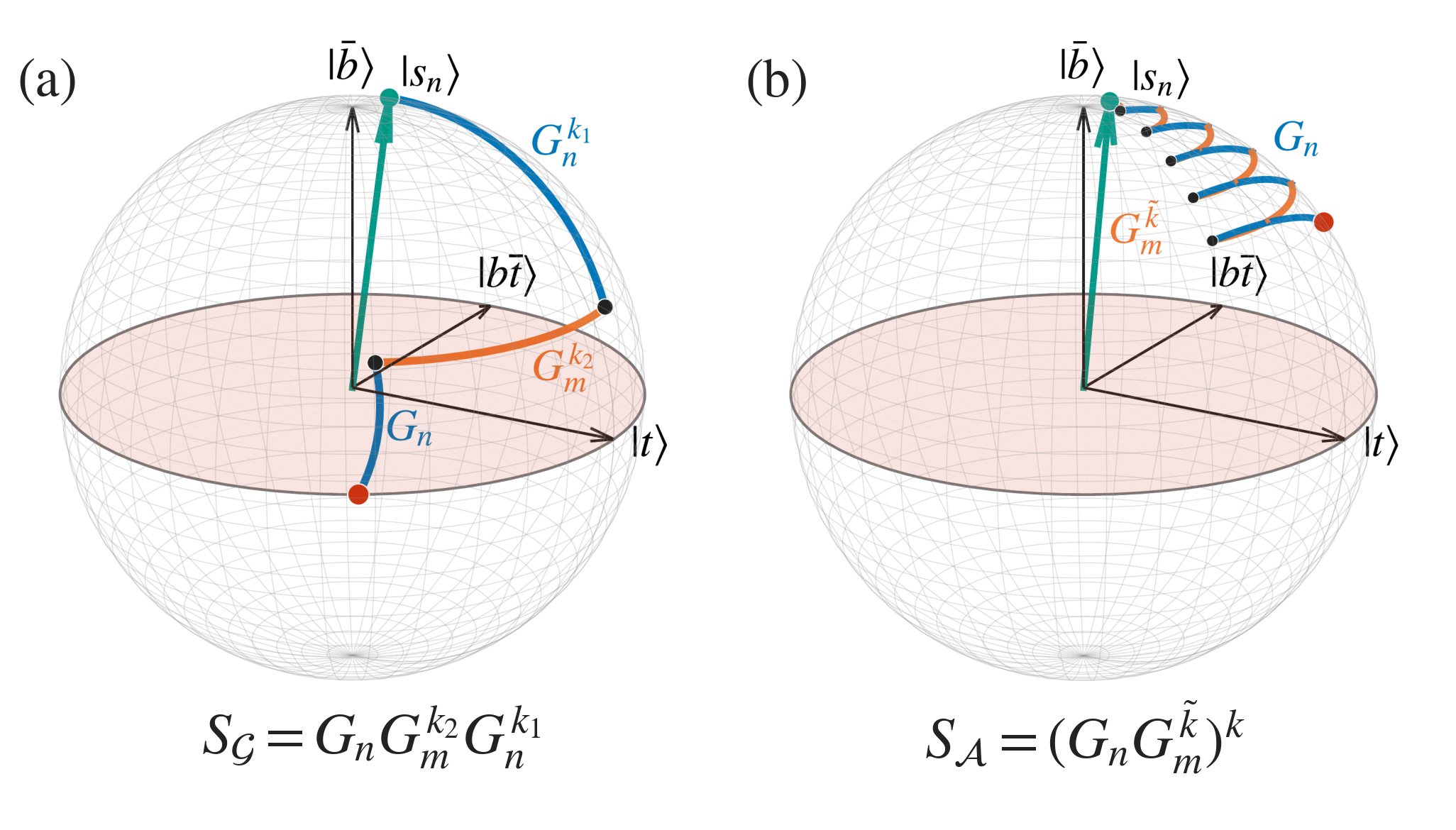}
\caption{Geometric illustration of the two families of partial-search sequences in the three-dimensional invariant subspace. (a) The GRK sequence $S_{\mathcal{G}}=G_nG_m^{k_2}G_n^{k_1}$ follows a global-local-global trajectory. The blue arcs denote global Grover iterations $G^{k_1}_n$, whereas the orange arc denotes the local Grover iterations $G_m^{k_2}$. (b) The depth-efficient alternating sequence $S_{\mathcal{A}}=(G_nG_m^{\tilde k})^k$ repeatedly applies the same local--global pair (orange and blue), producing a spiral trajectory. The coordinate axes mark the target item $|t\rangle$, the non-target-in-block state $|b\bar t\rangle$, and the non-target-block state $|\bar b\rangle$.}
\label{fig:grk_dep_opt_trajectory}
\end{figure}

The GRK sequence applies $k_1$ global iterations, then $k_2$ local iterations, and finally one global iteration, and is given by 
\begin{equation}
\label{eq:GRK_sequence}
S_{\mathcal{G}}(k_1,k_2)=G_nG_m^{k_2}G_n^{k_1}.
\end{equation}
Writing $\tilde{\omega}_1$ for the measured block-index string, success means $\tilde{\omega}_1=t_1$. For a fixed number of oracle queries $k_{\mathrm{tot}}=1+\beta\sqrt N$ with $0<\beta<\pi/4$, the maximum success probability within the GRK sequence has the following expansion in powers of $K^{-1/2}$ \cite{jiang2026exactboundsquantumpartial}
\begin{multline}
\label{eq:GRK_success_expand}
\pr_{\tilde{\omega}_1=t_1}^{\max}(k_{\mathrm{tot}})
=\sin^2(2\beta)+\kappa_1K^{-1/2}\sin(4\beta)\\
+\mathcal{O}(K^{-1})+\mathcal{O}(N^{-1/2}),
\end{multline}
where $\kappa_1\approx0.6849$. The $\mathcal{O}(N^{-1/2})$ term records finite-$N$ corrections that are not controlled by powers of $K^{-1/2}$ alone. In the joint limit $b,K\to\infty$, $N^{-1/2}=K^{-1/2}b^{-1/2}=o(K^{-1/2})$, so both displayed remainders are smaller than the partial-search correction term. The leading term is the global-search contribution under the same query count, while the displayed $K^{-1/2}$ correction is the partial-search gain. The success probability of the GRK sequence reaches its maximum at $k_2=\pi\sqrt{b}/6$ \cite{korepin2005simple,korepin2005optimization}. See Fig.~\ref{fig:grk_dep_opt_trajectory}(a) for the geometric illustration of the GRK algorithm. Throughout this work, $\Theta(f)$ denotes a term of the same asymptotic order as $f$ up to constant factors, $\mathcal{O}(f)$ a term whose magnitude is asymptotically bounded above by a constant multiple of $f$, and $o(f)$ a term asymptotically negligible relative to $f$ (i.e., its ratio to $f$ tends to zero).

For the expected-query efficiency metric under independent repetitions, the GRK algorithm can be optimized similarly \cite{jiang2026exactboundsquantumpartial}. In the limits $b\gg1$ and $K\gg1$, its optimal total query count is
\begin{equation}
\label{eq:GRK_query_opt}
k_{\mathrm{tot}}\simeq \kappa_2\sqrt{N}-\kappa_3\sqrt{b},
\end{equation}
where $\kappa_2\approx0.5829$ and $\kappa_3\approx0.4969$, unless the boundary solution $k_{\mathrm{tot}}=1$ yields a lower expected cost. For $m\lesssim\lfloor n/2\rfloor$, the minimum expected number of oracle queries is
\begin{equation}
\label{eq:E^min_oracle_GRK}
\mathrm{E}^{\min}_{\tilde{\omega}_1=t_1}\simeq \kappa_0\sqrt{N}-\kappa_4\sqrt{b},
\end{equation}
where $\kappa_4\approx0.4054$. The first term is the original Grover cost. The second negative correction is the query saving from partial search. For $m>\lfloor n/2\rfloor$, however, the number of blocks is too small for this quantum advantage to persist, and the optimal expected cost approaches $\mathrm{E}_{\tilde{\omega}_1=t_1}^{\min}\simeq K=2^{n-m}$, which is the cost of random guessing. Thus, within the expected-query metric, the speedup of quantum partial search is not unconditional. Only when the number of blocks $K$ is sufficiently large can GRK outperform classical partial search. For readers' convenience, we summarize the universal constants, such as $\kappa_0$, obtained from optimization in Table~\ref{tab:constants}.

\begin{table}[t]
\centering
\footnotesize
\caption{Numerical constants used throughout the paper and their associated equations.}
\label{tab:constants}
\setlength{\tabcolsep}{3pt}
\begin{tabular*}{\columnwidth}{@{\extracolsep{\fill}}cllc@{}}
\toprule
Symbol & Value & Definition & Equation \\
\midrule
$y_0$ & $1.1656$ & Smallest positive root of $\tan y = 2y$ & \eqref{eq:E_min_Grover} \\
$\kappa_0$ & $0.6900$ & $y_0/(2\sin^2 y_0)$ & \eqref{eq:E_min_Grover} \\
$\kappa_1$ & $0.6849$ & GRK success-probability correction & \eqref{eq:GRK_success_expand} \\
$\kappa_2$ & $0.5829$ & GRK query-complexity coefficient & \eqref{eq:GRK_query_opt} \\
$\kappa_3$ & $0.4969$ & GRK query-complexity correction & \eqref{eq:GRK_query_opt} \\
$\kappa_4$ & $0.4054$ & Partial-search query saving & \eqref{eq:E^min_oracle_GRK} \\
$\kappa_5$ & $0.6200$ & $\pi\kappa_0/(6\kappa_2)$ & \eqref{eq:GRKdepth_small_q} \\
$\kappa_6$ & $1.0400$ & $3^{2/3}/2$ & \eqref{eq:Enew_expand} \\
\bottomrule
\end{tabular*}
\end{table}

\section{\label{sec:numerical} Numerical optimization for depth-optimal partial search}

In this section, we numerically identify the optimal partial-search sequence using the circuit-depth metric rather than the number of oracle queries. Sec.~\ref{subsec:setups} specifies the numerical setup, and Sec.~\ref{subsec:results} presents the phase diagram of the depth-optimal partial-search sequences. 

\subsection{\label{subsec:setups} Numerical setup}

From the perspective of oracle-query complexity, global and local Grover operators are equivalent, with each counted as one query. From the perspective of circuit depth, however, the global Grover operator is more expensive than the local one because the local diffusion operator has lower circuit depth than the global diffusion operator. Both numerical and analytical results show that the GRK sequence, given by Eq.~\eqref{eq:GRK_sequence}, is strictly optimal in the number of oracle queries \cite{jiang2026exactboundsquantumpartial,zhang2026asymptotic}. However, when the metric changes to circuit depth, the GRK sequence is not guaranteed to be optimal. 

To numerically determine the depth-optimal partial-search sequence, we consider all admissible alternating sequences of global and local Grover operators, denoted by
\begin{equation}
\label{eq:S_sequence}
S(k_1,k_2,\ldots,k_J)
=G_n^{k_1}G_m^{k_2}G_n^{k_3}\cdots G_n^{k_J}.
\end{equation}
Here $J$ is an odd integer. The exponents satisfy $k_1,k_2,\ldots,k_{J-1}\in\mathbb Z_{\ge1}$, and $k_J\in\mathbb Z_{\ge0}$. The choice $k_J=0$ allows the first operator acting on the state to be local. However, we always set $k_1\geq 1$ because the local Grover operator does not change the amplitude of the block-label qubits, so the sequence must end with a global Grover operator. Note that quantum operators act on state vectors from right to left.

The total number of oracle queries of sequence $S$ is
\begin{equation}
k_{\mathrm{tot}} = \sum_{p=1}^{J} k_p,
\end{equation}
and the probability of finding the target block after applying $S$ is
\begin{equation}
\pr_{\tilde{\omega}_1=t_1}(S)=1-\bigl|\langle\bar b|S|s_n\rangle\bigr|^2.
\end{equation}
Since this work focuses on quantum partial search, we may omit the subscript $\tilde{\omega}_1=t_1$ below when there is no ambiguity. 

Let $\ell(S)$ denote the circuit depth of operator $S$. For brevity, we write $\ell_{O_t} = \ell(O_t)$, $\ell_n = \ell(D_n)$, and $\ell_m = \ell(D_m)$ for the depths of the oracle, global diffusion, and local diffusion operators, respectively. In practice, we have $\ell_m<\ell_n$, and the diffusion depth increases with register size, $\ell_m<\ell_{m+1}$. Because the oracle depth is problem-dependent, we parameterize it by the depth ratio
\begin{equation}
\alpha=\frac{\ell_{O_t}}{\ell_n},
\end{equation}
which gives $\ell_{O_t} = \alpha \ell_n$. When $\alpha\gg1$, oracle calls dominate the depth, and the depth and oracle-query metrics become equivalent \cite{zhang2026palindromic}. At finite $\alpha$, however, the two metrics can be inequivalent. This motivates us to explore depth-optimal partial-search sequences, which are relevant to depth-constrained implementations. 

Specifically, under this model, the depth of $k$ iterations of the global Grover operator is
\begin{equation}
\ell(G_n^k) = k(\ell_n + \ell_{O_t}) = k(\alpha+1)\ell_n,
\end{equation}
while that of $k$ iterations of the local Grover operator is
\begin{equation}
\ell(G_m^k) = k(\ell_m + \ell_{O_t}) = k(\ell_m + \alpha\ell_n).
\end{equation}
For a sequence $S$ in Eq.~\eqref{eq:S_sequence}, where odd indices correspond to $G_n$ and even indices to $G_m$, the total circuit depth is
\begin{equation}
\ell(S) = \Bigl(\alpha k_{\mathrm{tot}} + \sum_{\text{odd } p} k_p\Bigr)\ell_n + \Bigl(\sum_{\text{even } p} k_p\Bigr)\ell_m.
\end{equation}
The expected depth of finding the target block is
\begin{equation}
\mathrm{E}_{\ell}(S)=\frac{\ell(S)}{\pr(S)},
\end{equation}
which is the expected depth per successful sample under independent repetitions. This metric captures circuit-depth cost rather than query count. It averages the quantum circuit-depth expenditure over successful outcomes but does not include the cost of verifying a measured block label, which uses classical resources. Since our cost model focuses on quantum resources, we omit the classical verification cost. 

The optimization problem for finding the partial-search algorithm with minimum expected depth is
\begin{equation}
\mathrm{E}_{\ell}^{\min}
=\min_{S}\mathrm{E}_{\ell}(S),
\label{eq:E_opt}
\end{equation}
where $S$ is chosen from all admissible sequences. It is expected that the GRK sequence gives $\mathrm{E}_{\ell}^{\min}$ when $\alpha$ is very large. It is therefore essential to determine how large the critical $\alpha$ is and, below it, which sequence yields a lower minimum expected circuit depth than the GRK sequence.

\subsection{\label{subsec:results} Optimal sequences and phase diagram}

We numerically solve the optimization problem in Eq.~\eqref{eq:E_opt} by enumerating every admissible sequence up to a cutoff in $k_{\mathrm{tot}}$. For the linear-depth model used in the numerical scan below, any admissible word of length $k_{\mathrm{tot}}$ containing at least one global iteration satisfies
\begin{equation}
\mathrm{E}_{\ell}\ge(\alpha n+m)k_{\mathrm{tot}}+(n-m).
\label{eq:enumeration_lower_bound}
\end{equation}
Thus no word with $k_{\mathrm{tot}}\ge\lceil(\mathrm{E}_{\mathrm{best}}-(n-m))/(\alpha n+m)\rceil$ can improve the current optimum, where $\mathrm{E}_{\mathrm{best}}$ denotes the smallest expected depth found during the enumeration. Therefore we can safely exclude any $S$ whose length is at or above $\lceil(\mathrm{E}_{\mathrm{best}}-(n-m))/(\alpha n+m)\rceil$.

The numerical search is performed for $N=2^{10}$, $m=2,\ldots,9$, and
$\alpha\in\{1,3,5,10,100\}$. The diffusion depths are set to scale linearly with their register sizes, namely $\ell_n=\xi n$ and $\ell_m=\xi m$, corresponding to an ancilla-free linear-depth model for multi-controlled Toffoli gates \cite{saeedi2013linear,he2017decompositions,zindorf2025efficient}. We set the common positive scale $\xi=1$ without loss of generality. Table~\ref{tab:n10_alpha1_unrestricted} lists the results for $\alpha=1$, where the GRK sequence is not optimal in depth. Figure~\ref{fig:expected_depth_vs_m} compares the optimized results with those from the GRK ansatz at $\alpha=1$. For $m=2,\ldots,5$, the optimized results are substantially shallower than the GRK sequence, demonstrating that expected-query efficiency within the GRK ansatz does not imply depth optimality. For $m\ge6$, the optima reduce to a single global iteration, as also reported in Ref.~\cite{jiang2026exactboundsquantumpartial}, indicating that the quantum partial-search sequences are no longer efficient. For other representative values of $\alpha\in\{3,5,10,100\}$, the results are summarized in Appendix~\ref{app:tables}. In particular, when $\alpha=100$, the depth-optimized sequences are all GRK sequences, showing that the query and depth optima coincide at large $\alpha$. 

\begin{table}[t]
\centering
\footnotesize
\caption{Depth-optimal sequences and the corresponding success probabilities and expected depths for $\alpha=1$, $n=10$, $\ell_n=n$, and $\ell_m=m$. We apply the cutoff $k_{\mathrm{tot}}\le24$.}
\label{tab:n10_alpha1_unrestricted}
\setlength{\tabcolsep}{4pt}
\begin{tabular}{@{}clcccc@{}}
\toprule
$m$ & Optimal sequence $S$ & $k_{\mathrm{tot}}$ & $\ell(S)$ & $\pr(S)$ & $\mathrm{E}_{\ell}^{\min}$ \\
\midrule
2 & $(G_n G_m)^{10}$ & 20 & 320 & 0.824 & 388.5 \\
3 & $G_n G_m^2 (G_n G_m)^8$ & 19 & 310 & 0.848 & 365.5 \\
4 & $G_n G_m^3 (G_n G_m^2)^5$ & 19 & 302 & 0.859 & 351.8 \\
5 & $G_n G_m^4 (G_n G_m^2)^4$ & 17 & 280 & 0.830 & 337.3 \\
6 & $G_n$ & 1 & 20 & 0.070 & 286.5 \\
7 & $G_n$ & 1 & 20 & 0.132 & 151.7 \\
8 & $G_n$ & 1 & 20 & 0.256 & 78.2 \\
9 & $G_n$ & 1 & 20 & 0.504 & 39.7 \\
\bottomrule
\end{tabular}
\end{table}

\begin{figure}[!t]
\centering
\includegraphics[width=\columnwidth]{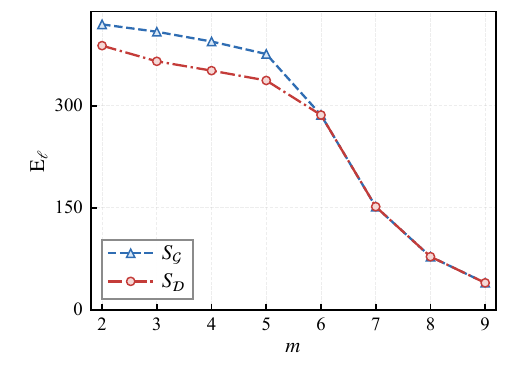}
\caption{Expected circuit depth $\mathrm{E}_{\ell}$ as a function of the block-size parameter $m$ for $N=2^{10}$ and $\alpha=1$. We assume the linear-depth model of diffusion operators with $\ell_n=n$ and $\ell_m=m$. Circles show the depth-optimal sequences listed in Table~\ref{tab:n10_alpha1_unrestricted}, which have the form $S_{\mathcal{D}}$ defined in Eq.~\eqref{eq:S_D}. Triangles show the optimal expected circuit depth restricted to the GRK sequence given by $S_{\mathcal{G}}=G_nG_m^{k_2}G_n^{k_1}$. }
\label{fig:expected_depth_vs_m}
\end{figure}

\begin{table}[t]
\centering
\footnotesize
\caption{Summary of the three sequence families used throughout this work.}
\label{tab:sequence_families}
\renewcommand{\arraystretch}{1.15}
\begin{tabular}{@{}c@{\hspace{1pt}}c@{\hspace{1pt}}l@{}}
\toprule
Sequence & Operator form & Comment \\
\midrule
${S}_{\mathcal{D}}(k,\bar k,\tilde k)$
& ~$G_nG_m^{\bar k}(G_nG_m^{\tilde k})^k$~
& \makecell[l]{General family of depth-optimal sequences} \\
\addlinespace[2pt]
${S}_{\mathcal{G}}(k_1,k_2)$
& $G_nG_m^{k_2}G_n^{k_1}$
& \makecell[l]{GRK sequence with\\$S_{\mathcal{G}}(k_1,k_2)=S_{\mathcal{D}}(k_1,k_2,0)$} \\
\addlinespace[2pt]
${S}_{\mathcal{A}}(k,\tilde k)$
& $(G_nG_m^{\tilde k})^k$
& \makecell[l]{Alternating sequence with\\$S_{\mathcal{A}}(k,\tilde k)=S_{\mathcal{D}}(k-1,\tilde{k},\tilde{k})$} \\
\bottomrule
\end{tabular}
\end{table}

\begin{figure}[!t]
\centering
\includegraphics[width=\linewidth]{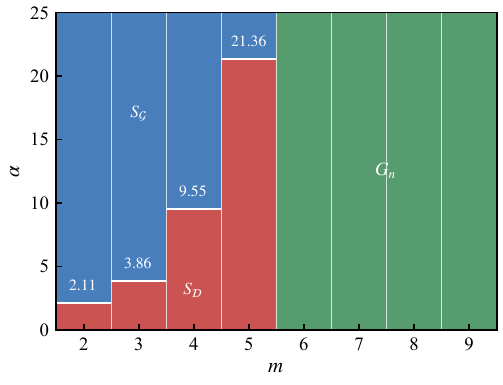}
\caption{Phase diagram of depth-optimal partial-search operators obtained from a numerical scan in $\alpha$ for the discrete block sizes $m=2,\ldots,9$ at $N=2^{10}$. We assume the linear-depth model of diffusion operators with $\ell_n=n$ and $\ell_m=m$. Red denotes the non-GRK $S_{\mathcal{D}}$ region, blue $S_{\mathcal{G}}$, and green the single global iteration $G_n$. The annotated crossings are $\alpha_{\mathrm{c}}=2.11,3.86,9.55,21.36$ for $m=2,3,4,5$.}
\label{fig:phase_diagram_v2}
\end{figure}

Across all numerically tested $\alpha$ values, we find that the depth-optimal sequences fall into the three-parameter family
\begin{equation}
\label{eq:S_D}
S_{\mathcal{D}}(k,\bar k,\tilde{k})=G_nG_m^{\bar k}(G_n G_m^{\tilde{k}})^k,
\end{equation}
with $k,\bar k,\tilde{k}\in\mathbb Z_{\ge 0}$.
The $m>\lfloor n/2\rfloor$ regime gives a random-guess solution with the operator $G_n = S_{\mathcal{D}}(0,0,0)$. For $m\le\lfloor n/2\rfloor$, the optimal sequences fall into two major categories. When $\alpha$ is large, the oracle dominates the depth and the GRK structure is selected, given by
\begin{equation}
S_{\mathcal{G}}(k_1,k_2)=S_{\mathcal{D}}(k_1,k_2,0)=G_nG_m^{k_2}G_n^{k_1}.
\end{equation}
When $\alpha$ is below a threshold, replacing global diffusions with local ones becomes advantageous, and the optimal sequence has the alternating form
\begin{equation}
\label{eq:S_A}
S_{\mathcal{A}}(k,\tilde k)=S_{\mathcal{D}}(k-1,\tilde{k},\tilde{k})=(G_n G_m^{\tilde{k}})^k,
\end{equation}
with $k\in\mathbb Z_{\ge1}$. The alternating sequence has a geometric interpretation distinct from that of the GRK sequence. See Fig.~\ref{fig:grk_dep_opt_trajectory}(b). For finite $N$, the optimal sequences need not contain an integer number of identical alternating units. This produces an edge operator $G_nG_m^{\bar k}$, as in $G_nG_m^{2}(G_nG_m)^8$ with $m=3$ and $G_nG_m^{3}(G_nG_m^{2})^{5}$ with $m=4$ in Table~\ref{tab:n10_alpha1_unrestricted}. However, most of the structure follows $S_{\mathcal{A}}(k,\tilde k)$. Table~\ref{tab:sequence_families} summarizes the three sequence families and their relations.

We present the phase diagram of the depth-optimal sequence in Fig.~\ref{fig:phase_diagram_v2}. It shows three distinct regimes: $G_n$, the GRK sequence $S_{\mathcal{G}}$, and
$S_{\mathcal{D}}$ outside the GRK form. For the phase diagram, we scan $0\le\alpha\le25$ and solve each crossing between $S_{\mathcal{G}}$ and the non-GRK portion of the $S_{\mathcal{D}}$ family. This gives the critical ratios $\alpha_{\mathrm{c}}\approx 2.11$, $3.86$, $9.55$, and $21.36$ for $m=2,3,4,5$, respectively. For
$m=7,8,9$, $G_n$ occupies the entire scan range.

\section{\label{sec:analytical} Analytical theory of the alternating sequence}

We analyze the success probability and the minimum expected circuit depth of the alternating sequence $S_{\mathcal{A}}(k,\tilde k)$ defined in Eq.~\eqref{eq:S_A} in Secs.~\ref{subsec:operator} and~\ref{subsec:expected-depth}, respectively. In particular, the optimization of the minimum expected depth is divided into three cases corresponding to distinct parameter regimes of the depth ratio $\alpha$.


\subsection{\label{subsec:operator} Success probability of the alternating sequence}

The alternating sequence $S_{\mathcal{A}}(k,\tilde k)$ defined in Eq.~\eqref{eq:S_A} consists of repeated applications of $G_nG_m^{\tilde k}$, which acts as an orthogonal transformation on the three-dimensional partial-search subspace spanned by $\{|t\rangle,|b\bar t\rangle,|\bar b\rangle\}$ defined in Eq.~\eqref{eq:three_dimen_basis}. On the initial state $|s_n\rangle$, this transformation realizes a rotation within the two-dimensional invariant subspace containing the initial state. Specifically, the success probability of $S_{\mathcal{A}}(k,\tilde k)$ has the following closed form.

\begin{lemma}\label{lem:exactProb}
For $k\in\mathbb Z_{\ge0}$ and $\tilde k\in\mathbb Z_{\ge0}$, the probability that $S_{\mathcal{A}}(k,\tilde{k})=(G_nG_m^{\tilde{k}})^k$ finds the target block is
\begin{multline}
\pr_{\mathcal{A}}(k,\tilde{k})
=1-\frac{K-1}{K-\sin^2\lambda}\\
\times\cos^2\!\left((2k+1)\arcsin\!\left(\frac{\sin\lambda}{\sqrt K}\right)\right),
\label{eq:exactP}
\end{multline}
where $\lambda=(\tilde{k}+1)\theta_2$ and $\theta_2=\arcsin(1/\sqrt b)$.
\end{lemma}

Here, $\arcsin(\cdot)$ denotes the principal inverse sine, whose range is $[-\pi/2,\pi/2]$. To exclude depth-inefficient overrotation, we explicitly restrict the local rotation to its first physical interval, $0<\lambda\le\pi/2$. The operator $G_nG_m^{\tilde k}\in O(3)$ has determinant $-1$. The initial state $|s_n\rangle$ is orthogonal to the $(-1)$ eigenspace. Analyzing the resulting two-dimensional rotation yields Eq.~\eqref{eq:exactP}. The full derivation is given in Appendix~\ref{app:recurrence}. The success-probability expression in Eq.~\eqref{eq:exactP} is exact. However, the leading-order approximation is more convenient for determining the leading scaling of the algorithm's efficiency.

\begin{corollary}\label{lem:asympProb}
For $K\gg1$ and $b\gg1$, assume
$k=\mathcal{O}(\sqrt K/\sin\lambda)$. The target-block success probability of $S_{\mathcal{A}}(k,\tilde{k})=(G_nG_m^{\tilde{k}})^k$ is
\begin{multline}
\pr_{\mathcal{A}}(k,\tilde{k})
=\sin^2\!\left(\frac{2k+1}{\sqrt K}\sin\!\left(\frac{\tilde{k}+1}{\sqrt b}\right)\right)
\\
+\mathcal{O}(K^{-1})+\mathcal{O}(b^{-1}).
\label{eq:approxP2}
\end{multline}
\end{corollary}

\begin{proof}
Starting from Lemma~\ref{lem:exactProb}, for $K\gg1$, the prefactor satisfies
\begin{equation}
\frac{1-K^{-1}}{1-K^{-1}\sin^2\lambda}
=1-K^{-1}\cos^2\lambda+\mathcal{O}(K^{-2}).
\end{equation}
Therefore,
\begin{equation}
\pr_{\mathcal{A}}(k,\tilde k) = \sin^2\!\left((2k+1)\arcsin\!\left(\frac{\sin\lambda}{\sqrt K}\right)\right) + \mathcal{O}(K^{-1}).
\label{eq:approxP1}
\end{equation}
For large $K$, the inverse sine has the expansion
\begin{equation}
\arcsin\!\left(\frac{\sin\lambda}{\sqrt K}\right)
=\frac{\sin\lambda}{\sqrt K}
+\mathcal{O}\!\left(\frac{\sin^3\lambda}{K^{3/2}}\right).
\end{equation}
Under the stated condition on $k$, namely $k=\mathcal{O}(\sqrt K/\sin\lambda)$, the accumulated phase error is
$(2k+1)\mathcal{O}(\sin^3\lambda/K^{3/2})=\mathcal{O}(\sin^2\lambda/K)=\mathcal{O}(K^{-1})$.
Moreover, $\theta_2=b^{-1/2}+\mathcal{O}(b^{-3/2})$. On $0<\lambda\le\pi/2$, where $\sin\lambda$ is comparable to $\lambda=(\tilde{k}+1)\theta_2$, replacing $(\tilde{k}+1)\theta_2$ by $(\tilde{k}+1)/\sqrt b$ changes the accumulated phase by $\mathcal{O}(b^{-1})$. These estimates yield Eq.~\eqref{eq:approxP2}.
\end{proof}

The condition $k=\mathcal{O}(\sqrt K/\sin\lambda)$ covers the natural amplitude-amplification window. One alternating unit produces an effective rotation of scale $\Theta(\sin\lambda/\sqrt K)$, so a total rotation of scale $\Theta(1)$ is reached when $k=\Theta(\sqrt K/\sin\lambda)$. Smaller iteration counts are already included in the stated $\mathcal{O}(\sqrt K/\sin\lambda)$ range. The condition therefore covers the nontrivial optimization region while excluding much longer sequences for which the higher-order phase error can accumulate beyond the stated sub-leading order. In particular, the value of $k$ that minimizes the expected depth, as shown below, satisfies the condition.

\subsection{\label{subsec:expected-depth} Minimum expected depth of the alternating sequence}

We minimize the expected depth of the partial-search algorithm implemented by the alternating sequence $S_{\mathcal{A}}(k,\tilde k)=(G_nG_m^{\tilde k})^k$ over $k\in\mathbb Z_{\ge1}$ and $\tilde k\in\mathbb Z_{\ge0}$. The optimization is performed in the large-$K$, large-$b$ regime and is understood as a leading asymptotic optimization followed by an integer rounding of $k$ and $\tilde k$.

One alternating unit contains one global Grover operator and $\tilde k$ local Grover operators. Hence
\begin{equation}
\ell(S_{\mathcal{A}})=k\left((\alpha\ell_n+\ell_m)\tilde k+(\alpha+1)\ell_n\right).
\label{eq:depth}
\end{equation}
The asymptotic success probability $\pr_{\mathcal{A}}(k,\tilde{k})$ in Eq.~\eqref{eq:approxP1} gives the expected depth at leading order, namely
\begin{multline}
\mathrm{E}_{\ell,\mathcal{A}}(k,\tilde{k})
= \frac{k\bigl((\alpha\ell_n + \ell_m)\tilde{k} + (\alpha+1)\ell_n\bigr)}
{\sin^2\!\left( \dfrac{2k+1}{\sqrt K}\,\sin\left((\tilde{k}+1)\theta_2\right) \right)}\\
\times \bigl(1 + \mathcal{O}(K^{-1})\bigr).
\label{eq:Eapprox}
\end{multline}
Note that the depths of one global and one local Grover operator are $\ell_{G_n} = (\alpha+1)\ell_n$ and $\ell_{G_m} = \alpha\ell_n + \ell_m$, respectively.

We introduce the diffusion-depth saving parameter
\begin{equation}
\ell_\Delta = \ell_n - \ell_m,
\end{equation}
which equals the depth saved by replacing one global Grover operator $G_n$ with a local Grover operator $G_m$. For simplicity, we also denote
\begin{equation}
y = \frac{2k+1}{\sqrt K}\,|\sin(x\theta_2)|,\label{eq:intro_y}
\end{equation}
where $x=\tilde{k}+1$. Substituting the definitions $\ell_\Delta$, $x$, and $y$ into Eq.~\eqref{eq:Eapprox} gives
\begin{multline}
\mathrm{E}_{\ell,\mathcal{A}}(x,y)
= \frac{\sqrt K}{2}\,\frac{y}{\sin^2 y}\,
\frac{\ell_{G_m} x + \ell_\Delta}{|\sin(x\theta_2)|}\\
\times\Bigl(1 + \mathcal{O}\bigl(K^{-1/2}\bigr)\Bigr).
\label{eq:E_xy}
\end{multline}
The dependence on $x$ and $y$ factorizes. For fixed $x$, minimizing $y/\sin^2y$ determines the stopping iteration count $k$. The optimal inner iteration count $\tilde k$ is found by minimizing
\begin{equation}
\label{eq:h(x)}
h(x)=\frac{\ell_{G_m}x+\ell_\Delta}{|\sin(x\theta_2)|},
\end{equation}
over admissible integer values of $x=\tilde k+1$ for which the denominator is nonzero. The absolute value is required because the success probability depends on $\sin^2(x\theta_2)$. By folding $x\theta_2$ modulo $\pi$ and reflecting about $\pi/2$, any point outside the first positive quarter-period has a counterpart in $0<x\le\pi/(2\theta_2)$ with a smaller value of $h(x)$. Hence we can restrict the optimization to $\sin(x\theta_2)>0$.

The minimization over $y$ is straightforward. Treating $y$ as a continuous variable, $y/\sin^2y$ is minimized at $y_0\approx1.1656$, the smallest positive root of $\tan y=2y$. This is also the Grover stopping point given by Eq.~\eqref{eq:E_min_Grover}.

\begin{lemma}\label{lem:opt_k}
In the limit $K\gg1$, the optimal outer iteration count $k$ for fixed $x=\tilde{k}+1$ is
\begin{equation}
k_*(x)=\frac{y_0\sqrt{K}}{2\sin(x\theta_2)}-\frac12+\mathcal{O}(1).
\label{eq:k_opt_x}
\end{equation}
\end{lemma}

A detailed proof is presented in Appendix~\ref{app:proof_optimal}. The $\mathcal{O}(1)$ term accounts for integer rounding and finite-$K$ corrections, which change the optimal iteration count $k$ by only a bounded amount. Relative to the leading $\sqrt K/\sin(x\theta_2)$ term, it is $\mathcal{O}(K^{-1/2})$. The factor $\sin(x\theta_2)$ determines the effective rotation per application of $G_nG_m^{\tilde{k}}$. A larger rotation factor reduces the required number $k$ of applications of $G_nG_m^{\tilde{k}}$. When $\tilde{k}=0$ ($x=1$), $\sin\theta_2=1/\sqrt b$, so the sequence is ordinary global Grover search and $k_*\approx y_0\sqrt N/2$.

The minimization over $x$ determines whether local Grover operators provide a net depth advantage. The boundary value $x=1$ corresponds to $\tilde k=0$, with no local operations inside $G_nG_m^{\tilde{k}}$. Define
\begin{equation}
q_{\mathcal{A}}=\frac{\ell_\Delta}{\ell_{G_m}}.
\label{eq:qA_definition}
\end{equation}
This ratio compares the diffusion-depth saving $\ell_\Delta$ with the depth $\ell_{G_m}$ of one local Grover operator. Under the linear-depth specialization $\ell_n=n$ and $\ell_m=m$, it becomes $q_{\mathcal{A}}=(n-m)/(\alpha n+m)$ and therefore compactly encodes how $n$, $m$, and $\alpha$ determine the tradeoff. We next determine the optimal iteration parameters in three asymptotic regimes, which correspond to those identified in Fig.~\ref{fig:phase_diagram_v2}. Cases~I and~II are formulated in terms of $q_{\mathcal{A}}$ and $b=2^m$. Case~III instead arises when $K=2^{n-m}$ becomes too small to retain the quantum search advantage and is stated directly in terms of $m-n/2$ and $\alpha$.

\subsubsection{Case I: Cost-effective local-iteration regime}

Figure~\ref{fig:phase_diagram_v2} shows that when $\alpha$ is below a threshold and $m\leq\lfloor n/2\rfloor$, the alternating sequence $S_{\mathcal A}$ is the depth-optimal operator. We call this Case~I. In this regime, local Grover operators are sufficiently cost-effective that additional local iterations are beneficial. Mathematically, minimizing $h(x)$ in Eq.~\eqref{eq:h(x)} produces a finite optimal value of $\tilde k$, the inner local-iteration count. The numerator of $h(x)$ grows linearly with $x$, while the denominator initially grows through the effective rotation factor $\sin(x\theta_2)$. Their balance produces an interior optimum. The optimal iteration counts $k$ and $\tilde k$ are given below.

\begin{theorem}\label{thm:MED_caseI}
Assume $K\gg1$, $b\gg1$, and $m<n/2+\mathcal{O}(1)$. If
$q_{\mathcal{A}}b\to\infty$ and $q_{\mathcal{A}}/\sqrt b\to0$, equivalently
$b^{-1}\ll q_{\mathcal{A}}\ll b^{1/2}$, then the optimal iteration counts are given by
\begin{equation}
\tilde{k}_*=(3q_{\mathcal{A}})^{1/3}\theta_2^{-2/3}
-\frac{2}{5}q_{\mathcal{A}}-1
+\mathcal{O}\left(q_{\mathcal{A}}^{5/3}\theta_2^{2/3}\right),
\label{eq:k_tilde_opt_I}
\end{equation}
\begin{multline}
k_*=\frac{y_0\sqrt K}{2(3q_{\mathcal{A}}\theta_2)^{1/3}}
\left(1+\frac{3^{5/3}}{10}(q_{\mathcal{A}}\theta_2)^{2/3}\right.\\
\left.+\mathcal{O}\bigl((q_{\mathcal{A}}\theta_2)^{4/3}\bigr)\right)
-\frac12+\mathcal{O}(1).
\label{eq:k_opt_I}
\end{multline}
Here $y_0\approx1.1656$.
\end{theorem}

The proof and calculations are presented in Appendix~\ref{app:proof_optimal}. The theorem shows that the optimal iteration counts scale as $\tilde{k}_*=\Theta\bigl((q_{\mathcal{A}}b)^{1/3}\bigr)$ and $k_*=\Theta\bigl(N^{1/2}(q_{\mathcal{A}}b)^{-1/3}\bigr)$. When $q_{\mathcal{A}}=\Theta(1)$, we have $\tilde{k}_*=\Theta(b^{1/3})$ and $k_*=\Theta(K^{1/2}b^{1/6})=\Theta(N^{1/2}b^{-1/3})$. Hence the global Grover iteration count is smaller than the ordinary Grover scale $\Theta(N^{1/2})$ by a factor $b^{1/3}$. Meanwhile, $k_*\tilde{k}_*=\Theta(N^{1/2})$, so the depth advantage comes from replacing expensive global iterations with cheaper local ones while retaining a quadratic speedup. From these depth-optimal iteration counts, we obtain the corresponding minimum expected depth in Case~I.

\begin{corollary}\label{cor:MED_caseI_depth}
Under the assumptions of Theorem~\ref{thm:MED_caseI}, with the depth-optimal iteration counts $\tilde{k}_*$ and $k_*$ in Eqs.~\eqref{eq:k_tilde_opt_I} and~\eqref{eq:k_opt_I}, the minimum expected circuit depth is
\begin{multline}
\mathrm{E}_{\ell,\mathcal{A}}^{\min,\mathrm{I}}
=\kappa_0\ell_{G_m}\frac{2^{n/2}}
{\cos\left((\tilde{k}_*+1)\theta_2\right)}\\
\times\Bigl(1+\mathcal{O}(K^{-1/2})+\mathcal{O}(b^{-1})\Bigr),
\label{eq:MED_I}
\end{multline}
where $\kappa_0\approx0.6900$.
\end{corollary}

In Case~I, $(\tilde{k}_*+1)\theta_2=\Theta\bigl(q_{\mathcal{A}}^{1/3}b^{-1/6}\bigr)\to0$, so the cosine factor in Eq.~\eqref{eq:MED_I} approaches unity. The leading-order scaling is therefore $\mathrm{E}_{\ell,\mathcal{A}}^{\min,\mathrm{I}}
=\Theta\bigl(\ell_{G_m}\sqrt N\bigr)$. Thus, after factoring out the depth $\ell_{G_m}$ of one local Grover operator, the minimum expected depth retains the Grover scaling $N^{1/2}$. Since $N=Kb$, the same leading term scales as $K^{1/2}$ in the number of blocks and as $b^{1/2}$ in the block size. Under the linear-depth model with $\ell_n=n$ and $\ell_m=m$, we have $\ell_{G_m}=\alpha\log_2N+\log_2b$. The full leading-order scaling is then $\Theta\bigl((\alpha\log_2N+\log_2b)\sqrt{Kb}\bigr)$. Hence the Case~I advantage changes the operator-depth prefactor relative to global Grover search while preserving the square-root scaling in $N$.

For later comparison with GRK, expanding Eq.~\eqref{eq:MED_I} under the linear-depth model gives
\begin{multline}
\mathrm{E}_{\ell,\mathcal{A}}^{\min,\mathrm{asy}}
\simeq \kappa_0 n 2^{n/2}
\Biggl(\alpha+\frac{m}{n}\\
+\kappa_6\left(\frac{n-m}{n}\right)^{2/3}
\left(\alpha+\frac{m}{n}\right)^{1/3}2^{-m/3}\Biggr),
\label{eq:Enew_expand}
\end{multline}
where $\kappa_6=3^{2/3}/2\approx1.0400$. For $q_{\mathcal{A}}=\Theta(1)$, the omitted relative terms are $\mathcal{O}(2^{-2m/3})+\mathcal{O}(K^{-1/2})+\mathcal{O}(b^{-1})=o(2^{-m/3})$.

Case~I covers the full window $b^{-1}\ll q_{\mathcal{A}}\ll b^{1/2}$, with $q_{\mathcal{A}}=\Theta(1)$ representing a useful constant-order specialization. Under the linear-depth model, the conditions become
\begin{equation}
\frac{(n-m)2^m}{\alpha n+m}\gg1,
\qquad
\frac{n-m}{(\alpha n+m)2^{m/2}}\ll1.
\end{equation}
In terms of the parameter $\alpha$, the inequalities give $(n-m)2^{-m/2}/n\ll\alpha\ll(n-m)2^m/n$. The upper scale grows as $2^m$, while the lower scale decreases as $2^{-m/2}$. The admissible $\alpha$ interval is broad for sufficiently large $m$. Increasing $\alpha$ toward the upper scale $2^m$ makes the oracle contribution dominate the depth of each local iteration. Additional local diffusion operators then become depth-inefficient, which drives the optimized iteration count $\tilde k$ toward $0$.

\subsubsection{Case II: Local-iteration cost-effectiveness boundary}

As $q_{\mathcal{A}}$ decreases to order $b^{-1}$, the optimal $x=\tilde k+1$ approaches the first few admissible integers. The integer comparison must then be performed before taking a continuous expansion. Analytically, we find the boundary value of $q_{\mathcal{A}}$, below which inserting a local diffusion operator no longer reduces the expected depth.

\begin{theorem}\label{thm:MED_caseII}
Assume $K\gg1$, $b\gg1$, $m<n/2+\mathcal{O}(1)$, and $q_{\mathcal{A}}=\mathcal{O}(b^{-1})$. The exact integer optimum for $x$ is given by
\begin{equation}
x_{\mathrm{int}}
=\operatorname*{arg\,min}_{j\in\mathbb Z_{\ge1}}
\frac{\ell_{G_m}j+\ell_\Delta}{|\sin(j\theta_2)|}.
\label{eq:x_integer}
\end{equation}
If the optimized value is not unique, $x_{\mathrm{int}}$ denotes the smallest one. Define the boundary ratio by
\begin{equation}
\label{eq:q_bd}
q_{\mathrm{bd}}=\frac{2(1-\cos\theta_2)}{2\cos\theta_2-1}.
\end{equation}
If $q_{\mathcal{A}}\ge q_{\mathrm{bd}}$, we have $\tilde k_*=x_{\mathrm{int}}-1$, and $k_*$ is given by Eq.~\eqref{eq:k_opt_x}. For $q_{\mathcal{A}}<q_{\mathrm{bd}}$, the global integer optimum is $\tilde k_*=0$ and $k_*=y_0 2^{n/2}/2-1/2+\mathcal{O}(1)$.
\end{theorem}

The boundary value $q_{\mathrm{bd}}$ marks the first exact integer switch in $\tilde k$. It follows from $h(2)=h(1)$, with $\cos\theta_2=\sqrt{1-b^{-1}}$. Rounding the Case~I expression cannot locate the first switch. The proof instead compares the exact integer values and shows that $x=1$ is the global optimum below $q_{\mathrm{bd}}$. Further details of the proof and calculations are given in Appendix~\ref{app:proof_optimal}. The corresponding minimum expected depths follow directly from these depth-optimal iteration counts.

\begin{corollary}\label{cor:MED_caseII_depth}
Under the assumptions of Theorem~\ref{thm:MED_caseII}, if $q_{\mathcal{A}}\ge q_{\mathrm{bd}}$, the minimum expected circuit depth is
\begin{equation}
\mathrm{E}_{\ell,\mathcal{A}}^{\min,\mathrm{II}}
=\kappa_0\sqrt K\,
\frac{\ell_{G_m}x_{\mathrm{int}}+\ell_\Delta}{|\sin(x_{\mathrm{int}}\theta_2)|}
\bigl(1+\mathcal{O}(K^{-1/2})\bigr).
\label{eq:MED_boundary_layer}
\end{equation}
For $q_{\mathcal{A}}<q_{\mathrm{bd}}$, the alternating family reduces to a sequence of global Grover operators, and its minimum expected circuit depth is
\begin{equation}
\mathrm{E}_{\ell,\mathcal{A}}^{\min,\mathrm{II}}
=\kappa_0(\alpha+1)\ell_n2^{n/2}
\bigl(1+\mathcal{O}(2^{-n/2})\bigr).
\label{eq:MED_II}
\end{equation}
\end{corollary}

Under the linear-depth model, the first switch of $\tilde k$ from $0$ to $1$ occurs at
\begin{equation}
\alpha_{\mathrm{bd}}=\frac{n-m}{nq_{\mathrm{bd}}}-\frac{m}{n}
=\frac{n-m}{n}2^m-\frac{m}{n}+\mathcal{O}(1).
\label{eq:alpha_bd}
\end{equation}
Below $\alpha_{\mathrm{bd}}$, at least one local insertion reduces the expected depth. Above it, local operations cost more than they save, and the sequence collapses to ordinary global Grover search. Note that the critical depth ratio $\alpha_{\mathrm{bd}}$ only separates the alternating-sequence regime $S_{\mathcal A}$ from ordinary global Grover search. We find the more relevant critical depth ratio that separates the alternating sequence $S_{\mathcal A}$ from the GRK sequence in Sec.~\ref{subsec:alpha_c}.

\begin{figure}[t]
\centering
\includegraphics[width=\linewidth]{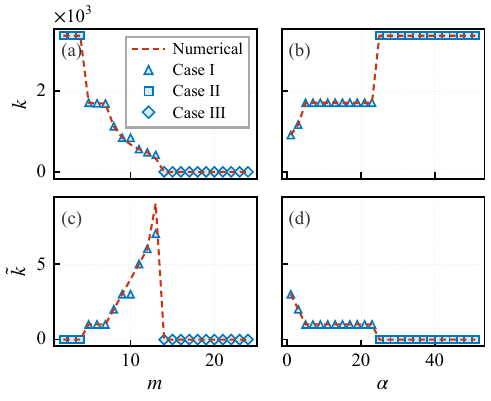}
\caption{Depth-optimal iteration counts for the alternating sequence $S_{\mathcal{A}}=(G_nG_m^{\tilde{k}})^k$ at $N=2^{25}$. Panels (a) and (b) show the outer iteration count $k$ as functions of $m$ at fixed $\alpha=20$ and of $\alpha$ at fixed $m=5$, respectively. Panels (c) and (d) show the corresponding inner iteration count $\tilde{k}$. Red dashed lines denote numerical results. Blue triangles, squares, and diamonds denote the analytical predictions in Cases~I, II, and III, given by Theorems~\ref{thm:MED_caseI}, \ref{thm:MED_caseII}, and~\ref{thm:MED_caseIII}, respectively.}
\label{fig:quantum_partial_search_2x2_n25}
\end{figure}

\subsubsection{Case III: Large-\texorpdfstring{$m$}{m} single-iteration solution}

Figure~\ref{fig:phase_diagram_v2} shows that as $m$ approaches $n/2$ from above, the quantum partial-search algorithm becomes depth-inefficient. The optimization then gives a single-iteration solution. In this regime, the number of blocks $K=2^{n-m}$ becomes too small for repeated amplitude amplification to retain a quantum advantage.

\begin{theorem}\label{thm:MED_caseIII}
Under the linear-depth model $\ell_n=n$ and $\ell_m=m$, in the large-$N$ limit with fixed $\alpha\ge0$, let $m=n/2+c$ with $c=\mathcal{O}(1)$ and define
\begin{equation}
c_*=\log_2\!\frac{\alpha+1}
{\kappa_0(\alpha+1/2)}.
\label{eq:caseIII_crossover}
\end{equation}
For $c>c_*+o(1)$, the depth-minimizing algorithm is given by $k_* = 1$ and $\tilde{k}_* = 0$.
The corresponding sequence is $S_{\mathcal A}(1,0)=G_n$. Its minimum expected circuit depth is
\begin{equation}
\mathrm{E}_{\ell,\mathcal{A}}^{\min,\mathrm{III}}
=(\alpha+1)\ell_n
\frac{K}{1+8(K-1)/N}
\left(1+\mathcal{O}\!\left(\frac{K}{N^2}\right)\right).
\label{eq:MED_III}
\end{equation}
\end{theorem}

The threshold $c_*$ follows by comparing the leading expected depth of one global iteration, namely $n(\alpha+1)2^{n/2-c}$, with the competing Case~I depth $\kappa_0n(\alpha+1/2)2^{n/2}$. Their equality gives Eq.~\eqref{eq:caseIII_crossover}. Numerically, $c_*(\alpha=1)\approx0.950$ and $c_*(\alpha\to\infty)\approx0.535$. Substituting $(k,\tilde k)=(1,0)$ into the exact probability in Lemma~\ref{lem:exactProb} gives Eq.~\eqref{eq:MED_III}. The loss of the repeated-search advantage is consistent with the query-complexity analysis of Ref.~\cite{jiang2026exactboundsquantumpartial}. Details of the proof and calculations are given in Appendix~\ref{app:proof_optimal}.

\begin{figure}[t]
\centering
\includegraphics[width=\linewidth]{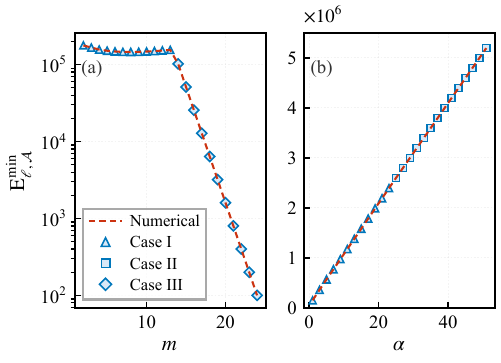}
\caption{Minimum expected circuit depth $\mathrm{E}_{\ell,\mathcal A}^{\min}$ of the alternating sequence for $N=2^{25}$. Panel (a) shows its dependence on $m$ at fixed $\alpha=1$ on a logarithmic vertical scale. Panel (b) shows its dependence on $\alpha$ at fixed $m=5$ on a linear vertical scale. Red dashed lines denote numerical results. Blue triangles, squares, and diamonds denote the analytical predictions in Cases~I, II, and III, given by Corollaries~\ref{cor:MED_caseI_depth}, \ref{cor:MED_caseII_depth}, and Theorem~\ref{thm:MED_caseIII}, respectively. We use the linear-depth model $\ell_n=n$ and $\ell_m=m$.}
\label{fig:min_E_n25_theory_vs_num}
\end{figure}

Figure~\ref{fig:quantum_partial_search_2x2_n25} compares the analytical and numerical depth-optimal iteration counts for $N=2^{25}$. At fixed $\alpha=20$, panels (a) and (c) pass through all three cases as $m$ increases. For $2\leq m\leq4$, the optimum is in Case~II with $\tilde k=0$, so the sequence consists only of global Grover operators. Case~I begins at $m=5$. Within $5\leq m\leq13$, the optimal $\tilde k$ increases in discrete steps while the corresponding outer iteration count $k$ decreases, showing that longer local subsequences replace some of the repeated global iterations. At $m=14$, the optimum enters Case~III with $(k,\tilde k)=(1,0)$. At fixed $m=5$, panels (b) and (d) instead show a Case~I--II transition as $\alpha$ increases. The optimal $\tilde k$ decreases from $3$ to $1$ within Case~I and becomes $0$ from $\alpha=25$, while $k$ increases to the value for global Grover search.

We compare the analytical and numerical results for the minimum expected depth in Fig.~\ref{fig:min_E_n25_theory_vs_num}. At fixed $\alpha=1$, panel (a) remains in Case~I for $2\leq m\leq13$, where the minimum expected depth changes only moderately. The threshold $m>n/2+c_*(1)\approx13.45$ from Eq.~\eqref{eq:caseIII_crossover} is first crossed at $m=14$, producing a direct transition to Case~III. In this regime, Eq.~\eqref{eq:MED_III} gives $\mathrm{E}_{\ell,\mathcal A}^{\min,\mathrm{III}}\approx(\alpha+1)\ell_nK$. Since $K=2^{n-m}$, each unit increase in $m$ approximately halves the minimum expected depth. Case~II does not occur in panel (a) because local iterations remain cost-effective at $\alpha=1$ before the Case~III threshold is reached. At fixed $m=5$, panel (b) remains in Case~I through $\alpha=23$ and enters Case~II at $\alpha=25$, consistent with $\alpha_{\mathrm{bd}}\approx24.4$ from Eq.~\eqref{eq:alpha_bd}.

\section{\label{sec:critical} Critical depth ratios relative to the GRK algorithm}

This section compares the depth-optimized GRK and alternating sequences $S_\mathcal{A}$. In Secs.~\ref{subsec:GRKdepth}, \ref{subsec:alpha_c}, and~\ref{subsec:DR}, respectively, we optimize the depth of the partial-search algorithm based on the GRK ansatz, derive the critical depth ratio, and evaluate the reduction ratio in the minimum expected circuit depth relative to the GRK algorithm. 

\subsection{\label{subsec:GRKdepth} Minimum expected depth of the GRK algorithm}

The GRK algorithm with the sequence $S_{\mathcal{G}}(k_1,k_2)=G_nG_m^{k_2}G_n^{k_1}$ is nontrivial for $m\le\lfloor n/2\rfloor$. For larger
$m$, its expected-query advantage over classical search disappears \cite{jiang2026exactboundsquantumpartial}. Set the rescaled iteration counts for the GRK algorithm as
\begin{equation}
v_{\mathrm{tot}}=\frac{k_{\mathrm{tot}}}{\sqrt N},\qquad
v_2=\frac{k_2}{\sqrt b}.
\label{eq:GRK_rescaled_iterations}
\end{equation}
Using $K^{-1/2}=\sqrt{b/N}$, the GRK success probability within the
first physical interval has the expansion
\begin{multline}
\pr_{\mathcal{G}}(v_{\mathrm{tot}},v_2)
=\sin^2(2v_{\mathrm{tot}})\\
+\frac{2}{\sqrt K}\bigl(\sin(2v_2)-v_2\bigr)\sin(4v_{\mathrm{tot}})\\
+\mathcal{O}(K^{-1})+\mathcal{O}(N^{-1/2}).
\label{eq:GRK_prob_depthopt}
\end{multline}
Note that the first term comes from global search, while the second term represents the probability advantage of partial search. 

In the linear-depth model with $\ell_n=n$ and $\ell_m=m$, the single-run depth of the GRK sequence $S_{\mathcal{G}}(k_1,k_2)$ is
\begin{multline}
\ell(S_{\mathcal{G}})
=(\alpha+1)n k_{\mathrm{tot}}-(n-m)k_2\\
=(\alpha+1)n v_{\mathrm{tot}}\sqrt N-(n-m)v_2\sqrt b.
\label{eq:GRK_single_depth}
\end{multline}
The corresponding expected depth is therefore
\begin{equation}
\mathrm{E}_{\ell,\mathcal{G}}(v_{\mathrm{tot}},v_2)=\frac{(\alpha+1)n v_{\mathrm{tot}}\sqrt N-(n-m)v_2\sqrt b}
{\pr_{\mathcal{G}}(v_{\mathrm{tot}},v_2)}.
\label{eq:GRK_expected_depth}
\end{equation}
We then optimize the total iteration count $v_{\mathrm{tot}}$ and the local iteration count $v_2$ to obtain the minimum expected depth within the GRK sequence.

\begin{lemma}\label{lem:GRK_opt_parameters}
Assume $b,K\to\infty$, and $m\le\lfloor n/2\rfloor$. Define
\begin{equation}
q_{\mathcal{G}}=\frac{n-m}{(\alpha+1)n}. 
\label{eq:qG_definition}
\end{equation}
Within the continuous asymptotic
optimization of the GRK sequence, the minimizing parameters are given by
\begin{equation}
v_{\mathrm{tot},*}=\kappa_2+\mathcal{O}(K^{-1/2}),\quad v_{2,*}=\frac12\arccos\!\left(\frac{1-q_{\mathcal{G}}}{2}\right),
\label{eq:GRK_v2star}
\end{equation}
where $\kappa_2\approx0.5829$.
\end{lemma}

The detailed proof is presented in Appendix~\ref{app:critical}. Here $q_{\mathcal{G}}$ is the fractional depth saved by replacing a global Grover iteration with a local one. In contrast, $q_{\mathcal{A}}$ in Eq.~\eqref{eq:qA_definition} normalizes the same saving by the local-operator depth $\ell_{G_m}$ rather than the global-operator depth $\ell_{G_n}$ and satisfies $q_{\mathcal{G}}=q_{\mathcal{A}}/(1+q_{\mathcal{A}})$. The optimal global iteration count $k_{\mathrm{tot}}$ approaches the optimal Grover stopping
point. The local optimum depends on the relative depth saving
$q_{\mathcal{G}}$. It approaches $\pi/6$ as $q_{\mathcal{G}}\to0$ and increases
toward $\pi/4$ as $q_{\mathcal{G}}\to1$. Based on these optimal iteration counts, we obtain the following minimum expected depth for the GRK algorithm. 

\begin{theorem}\label{thm:GRKdepth}
Under the assumptions of Lemma~\ref{lem:GRK_opt_parameters}, define
\begin{equation}
\label{eq:g*}
g_*(q_{\mathcal{G}})=\sin(2v_{2,*})-(1-q_{\mathcal{G}})v_{2,*}.
\end{equation}
The minimum expected depth within the GRK sequence is
\begin{multline}
\mathrm{E}_{\ell,\mathcal{G}}^{\min}=\kappa_0(\alpha+1)n\sqrt N
-\frac{(\alpha+1)n}{\sin^2y_0}g_*(q_{\mathcal{G}})\sqrt b \\
+o((\alpha+1)n\sqrt b).
\label{eq:GRKdepth_expand}
\end{multline}
\end{theorem}

Here $g_*(q_{\mathcal{G}})$ is the optimized coefficient of the order-$\sqrt b$ correction, obtained by maximizing
$\sin(2v_2)-(1-q_{\mathcal{G}})v_2$ over $0<v_2<\pi/2$, with the maximum
attained at $v_{2,*}$. The first term in Eq.~\eqref{eq:GRKdepth_expand} comes from the full Grover search. The second, negative term of order $\sqrt b$ quantifies the depth saving produced by
local Grover iterations. See Appendix~\ref{app:depth_GRK} for the derivations. 

For later comparison, the small-$q_{\mathcal{G}}$ expansion, corresponding to a large depth ratio $\alpha$, is
\begin{equation}
g_*(q_{\mathcal{G}})=g_*(0)+(\pi/6)q_{\mathcal{G}}+\mathcal{O}(q_{\mathcal{G}}^2),   
\end{equation}
which gives
\begin{multline}
\mathrm{E}_{\ell,\mathcal{G}}^{\min}
=\kappa_0(\alpha+1)n\sqrt N\\
-\left(\kappa_4(\alpha+1)n+\kappa_5(n-m)
+\mathcal{O}\!\left(\frac{(n-m)^2}{(\alpha+1)n}\right)\right)\sqrt b\\
+o((\alpha+1)n\sqrt b),
\label{eq:GRKdepth_small_q}
\end{multline}
where $\kappa_4\approx0.4054$ and
$\kappa_5=\pi/(6\sin^2y_0)=\pi\kappa_0/(6\kappa_2)\approx0.6200$. Since the global Grover operator has depth $(\alpha+1)n$, Eq.~\eqref{eq:E^min_oracle_GRK} gives the first two terms
$\kappa_0(\alpha+1)n\sqrt N-\kappa_4(\alpha+1)n\sqrt b$ in
Eq.~\eqref{eq:GRKdepth_small_q}. The circuit-depth optimization further
accounts for the depth $n-m$ saved by each local iteration and produces the
additional term $-\kappa_5(n-m)\sqrt b$. This term modifies the coefficient of
the same order-$\sqrt b$ correction rather than entering at a smaller
asymptotic order.

\subsection{\label{subsec:alpha_c} Critical depth ratios}

In Sec.~\ref{subsec:expected-depth}, we have obtained the crossing point in the depth ratio $\alpha$ between the alternating sequence $S_\mathcal{A}$ and global Grover search, as given by Eq.~\eqref{eq:alpha_bd}. However, this transition does not include the GRK sequence, which is more depth-efficient than global Grover search, as indicated in Ref.~\cite{jiang2026exactboundsquantumpartial}. 

Under the linear-depth model, we have the depth-saving parameter $q_{\mathcal{A}}=(n-m)/(\alpha n+m)$ given by Eq.~\eqref{eq:qA_definition}. Recall that when $q_{\mathcal{A}}<q_{\mathrm{bd}}$, where $q_{\mathrm{bd}}$ is given by Eq.~\eqref{eq:q_bd}, the alternating sequence $S_\mathcal{A}$ reduces to global Grover search. Since the GRK sequence is more efficient than global Grover search in both oracle-query count and circuit depth, we expect $q_{\mathrm{bd}}<q_{\mathcal{A},\mathrm{c}}$, where $q_{\mathcal{A},\mathrm{c}}$ is the transition point between the alternating sequence $S_\mathcal{A}$ and the GRK sequence $S_\mathcal{G}$, which is uniquely determined by the following theorem. 

\begin{theorem}\label{thm:alpha_c}
Assume $b\gg1$, $K\gg1$, and $m\le\lfloor n/2\rfloor$. The critical cost ratio $q_{\mathcal{A},\mathrm{c}}$ is the unique solution of
\begin{multline}
(1+q_{\mathcal{A},\mathrm{c}})
\left[1-\frac{1}{\kappa_2}
g_*\!\left(\frac{q_{\mathcal{A},\mathrm{c}}}
{1+q_{\mathcal{A},\mathrm{c}}}\right)2^{-(n-m)/2}\right]\\
=\min_{j\in\mathbb Z_{\ge1}}
\frac{j+q_{\mathcal{A},\mathrm{c}}}
{\sqrt b\,|\sin(j\theta_2)|},
\label{eq:qAc_integer}
\end{multline}
where the function $g_*(q_{\mathcal{G}})$ with $q_{\mathcal{G}}=q_{\mathcal{A}}/(1+q_{\mathcal{A}})$ is defined in Eq.~\eqref{eq:g*}.
\end{theorem}

Equation~\eqref{eq:qAc_integer} follows by equating the depth-optimized GRK expansion in Eq.~\eqref{eq:GRKdepth_expand} with the minimum expected depth of the alternating sequence in Eq.~\eqref{eq:MED_boundary_layer}. The detailed derivation is given in Appendix~\ref{app:critical}. Once the critical $q_{\mathcal{A},\mathrm{c}}$ is obtained, the corresponding critical depth ratio is
\begin{equation}
\alpha_{\mathrm{c},\mathcal{A}}
=\frac{n-m}{nq_{\mathcal{A},\mathrm{c}}}-\frac{m}{n},
\label{eq:alpha_c_from_q}
\end{equation}
which satisfies $0<\alpha_{\mathrm{c},\mathcal{A}}<\alpha_{\mathrm{bd}}$, where $\alpha_{\mathrm{bd}}$ is given by Eq.~\eqref{eq:alpha_bd}. For $0<\alpha<\alpha_{\mathrm{c},\mathcal{A}}$, $S_{\mathcal{A}}$ is shallower than the GRK algorithm.

Consider $n,m\to\infty$ with $m\le n/2$ and $3m-n\to\infty$, so that $m$ lies asymptotically above the crossover $n/3$. Then, equating the minimum expected depths in Eqs.~\eqref{eq:Enew_expand} and~\eqref{eq:GRKdepth_expand} defines the asymptotic critical cost ratio $q_{\mathcal{A},\mathrm{c}}^{\mathrm{asy}}$ through
\begin{multline}
\kappa_6(q_{\mathcal{A},\mathrm{c}}^{\mathrm{asy}})^{-1/3}2^{-m/3}\\
+\frac{1+q_{\mathcal{A},\mathrm{c}}^{\mathrm{asy}}}
{\kappa_2q_{\mathcal{A},\mathrm{c}}^{\mathrm{asy}}}
g_*\!\left(\frac{q_{\mathcal{A},\mathrm{c}}^{\mathrm{asy}}}
{1+q_{\mathcal{A},\mathrm{c}}^{\mathrm{asy}}}\right)
2^{-(n-m)/2}=1.
\label{eq:qAc_implicit}
\end{multline}
The corresponding $\alpha_{\mathrm{c},\mathcal{A}}^{\mathrm{asy}}$ follows from Eq.~\eqref{eq:alpha_c_from_q} with $q_{\mathcal{A},\mathrm{c}}$ replaced by $q_{\mathcal{A},\mathrm{c}}^{\mathrm{asy}}$. The large-$m$ asymptotic solution derived below verifies the validity range of the root from its $n$ and $m$ scaling. When $3m-n$ is not asymptotically large and positive, Eq.~\eqref{eq:qAc_implicit} is only a formal continuation, and the integer equation in Theorem~\ref{thm:alpha_c} remains the appropriate comparison. At leading order, we can find the asymptotic critical cost ratio $q_{\mathcal{A},\mathrm{c}}^{\mathrm{asy}}$ and the corresponding depth ratio $\alpha_{\mathrm{c},\mathcal{A}}^{\mathrm{asy}}$ explicitly as follows. 

\begin{corollary}\label{cor:alpha_c_asymp}
Let $n,m\to\infty$ with $m\le\lfloor n/2\rfloor$. If $n-3m\to\infty$, the small-$m$ leading-order solution of Eq.~\eqref{eq:qAc_implicit} gives
\begin{equation}
q_{\mathcal{A},\mathrm{c}}^{\mathrm{asy}}
\simeq\frac{9}{8}2^{-m},\quad \alpha_{\mathrm{c},\mathcal{A}}^{\mathrm{asy}}
\simeq\frac{8}{9}\frac{n-m}{n}2^m-\frac{m}{n}.
\label{eq:alpha_c_left}
\end{equation}
If $3m-n\to\infty$, the large-$m$ leading-order solution gives
\begin{subequations}
\begin{align}
& q_{\mathcal{A},\mathrm{c}}^{\mathrm{asy}}
\simeq\frac{\kappa_4}{\kappa_0}2^{-(n-m)/2},\\
& \alpha_{\mathrm{c},\mathcal{A}}^{\mathrm{asy}}
\simeq\frac{\kappa_0}{\kappa_4}\frac{n-m}{n}2^{(n-m)/2}-\frac{m}{n}.
\label{eq:alpha_c_right}
\end{align}
\end{subequations}
\end{corollary}

Equation~\eqref{eq:alpha_c_left} gives $q_{\mathcal{A},\mathrm{c}}^{\mathrm{asy}}b\to9/8$, equivalently $q_{\mathcal{A},\mathrm{c}}^{\mathrm{asy}}\simeq(9/8)b^{-1}=\Theta(b^{-1})$. Thus, this root lies at the integer-boundary scale (Case II) rather than satisfying the interior condition $b^{-1}\ll q_{\mathcal{A},\mathrm{c}}^{\mathrm{asy}}\ll b^{1/2}$ (Case I). By contrast, Eq.~\eqref{eq:alpha_c_right} gives
\begin{subequations}
\begin{align}
&\frac{q_{\mathcal{A},\mathrm{c}}^{\mathrm{asy}}}{b^{-1}}
\simeq\frac{\kappa_4}{\kappa_0}2^{(3m-n)/2}\to\infty,\\
&\frac{q_{\mathcal{A},\mathrm{c}}^{\mathrm{asy}}}{b^{1/2}}
\simeq\frac{\kappa_4}{\kappa_0}2^{-n/2}\to0.
\end{align}
\end{subequations}
Thus the large-$m$ root lies in the interior regime $b^{-1}\ll q_{\mathcal{A},\mathrm{c}}^{\mathrm{asy}}\ll b^{1/2}$ as a consequence of the stated $n$ and $m$ limits. Moreover, we find that the critical depth ratio does not increase monotonically with $m$. It has the maximum shown below. 

\begin{corollary}\label{cor:alpha_c_max}
For the asymptotic crossover curve defined by Eq.~\eqref{eq:qAc_implicit}, the critical depth ratio at leading order has a maximum $\alpha_{\mathrm{c},\mathcal{A}}^{\max}\simeq0.2976\times 2^{n/3}$ at $m\simeq n/3+1.2175$. 
\end{corollary}

At this maximum, $q_{\mathcal{A},\mathrm{c}}^{\mathrm{asy}}b=\Theta(1)$, so it lies at the boundary scale rather than in the interior limit. Using $b=2^m$, Eq.~\eqref{eq:alpha_c_from_q} then gives $\alpha_{\mathrm{c},\mathcal{A}}^{\max}=\Theta(2^m)$. Since $m=n/3+\mathcal{O}(1)$ at the maximum, this is equivalently $\alpha_{\mathrm{c},\mathcal{A}}^{\max}=\Theta(2^m)=\Theta(2^{n/3})$, showing exponential growth with both $m$ and $n$ along the maximizing curve. The derivation of Eq.~\eqref{eq:qAc_implicit} and the proofs of Corollaries~\ref{cor:alpha_c_asymp} and~\ref{cor:alpha_c_max} are given in Appendix~\ref{app:critical}.

The results above concern the alternating family $S_{\mathcal{A}}$ compared with the GRK algorithm. For the optimized three-parameter family $S_{\mathcal{D}}$ defined in Eq.~\eqref{eq:S_D}, we denote its crossing with the GRK sequence $S_{\mathcal G}$ by $\alpha_{\mathrm{c},\mathcal{D}}$. The three-parameter family $S_{\mathcal{D}}$ can be viewed as an edge correction to the alternating family $S_{\mathcal{A}}$. Appendix~\ref{app:boundary_proof} shows that the edge operator gives a controlled subleading correction to the alternating family $S_{\mathcal{A}}$. Because $S_{\mathcal{A}}\subset S_{\mathcal{D}}$, the exact full-family transition of $S_{\mathcal{D}}$ cannot occur below the corresponding transition for $S_{\mathcal{A}}$. A closed-form analytical optimization of $S_{\mathcal{D}}$ is not available here, so we determine $\alpha_{\mathrm{c},\mathcal{D}}$ by numerical integer scans.

\begin{figure}[!t]
\centering
\includegraphics[width=\linewidth]{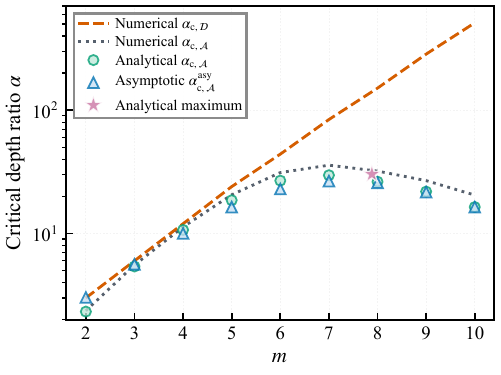}
\caption{Critical depth ratios as functions of $m$ for $N=2^{20}$. The dark-gray dotted line gives the numerical finite-size crossings $\alpha_{\mathrm{c},\mathcal{A}}$ obtained by exact integer optimization of $S_{\mathcal A}$ and GRK. The orange dashed line shows the numerical crossing points for the full-family transition $\alpha_{\mathrm{c},\mathcal{D}}$. Green circles are obtained from the reduced-integer analytical Eqs.~\eqref{eq:qAc_integer} and~\eqref{eq:alpha_c_from_q}, giving $\alpha_{\mathrm{c},\mathcal{A}}$. Blue triangles give the continuous-asymptotic analytical roots $\alpha_{\mathrm{c},\mathcal{A}}^{\mathrm{asy}}$ of Eq.~\eqref{eq:qAc_implicit}. The magenta star marks the large-$n$ analytical maximum in Corollary~\ref{cor:alpha_c_max}, evaluated at $n=20$. We use the linear-depth model $\ell_n=n$ and $\ell_m=m$.}
\label{fig:alpha_c_vs_m_aps}
\end{figure}

In Fig.~\ref{fig:alpha_c_vs_m_aps}, we numerically obtain the critical depth ratios of the optimized three-parameter family $S_{\mathcal{D}}$ and the alternating sequence $S_{\mathcal{A}}$. The numerical crossing $\alpha_{\mathrm{c},\mathcal A}$ first increases from $2.362$ at $m=2$ to a maximum of $35.81$ at $m=7$, and then decreases to $20.48$ at $m=10$. By contrast, the full-family crossing $\alpha_{\mathrm{c},\mathcal D}$ remains above $\alpha_{\mathrm{c},\mathcal A}$ and rises rapidly from $3$ to $511$ over the same range, showing that the finite-size edge degree of freedom can substantially delay the transition to the GRK family. The analytical points reproduce the nonmonotonic behavior of $\alpha_{\mathrm{c},\mathcal A}$, while the continuous-asymptotic roots approach the numerical results as $m$ increases. At $m=10$, the reduced-integer and continuous-asymptotic results give $16.388$ and $16.379$, respectively. The analytical maximum $(m,\alpha)\simeq(7.88,30.23)$ also locates the turnover near $m=8$. 

\subsection{\label{subsec:DR} Reduction ratio in the minimum expected circuit depth}

To quantify the efficiency of depth-optimized partial search, we consider the reduction ratio in the minimum expected circuit depth of the alternating sequence $S_\mathcal{A}$ relative to the GRK algorithm, given by 
\begin{equation}
\eta_{\mathcal A} = \frac{\mathrm{E}_{\ell,\mathcal{G}}^{\min} - \mathrm{E}_{\ell,\mathcal{A}}^{\min}}{\mathrm{E}_{\ell,\mathcal{G}}^{\min}} \times 100\%,
\end{equation}
where $\mathrm{E}_{\ell,\mathcal{G}}^{\min}$ and $\mathrm{E}_{\ell,\mathcal{A}}^{\min}$ are the minimum expected circuit depths of the GRK algorithm $S_\mathcal{G}$ and the alternating sequence $S_\mathcal{A}$, respectively. The corresponding ratio $\eta_{\mathcal D}$ for the three-parameter family $S_\mathcal{D}$ is defined similarly. 

For compactness, define the normalized correction parameters
\begin{align}
C_{\mathcal{A}}&=\kappa_6\left(\frac{n-m}{n}\right)^{2/3}
(\alpha+m/n)^{1/3}2^{-m/3},\notag\\
C_{\mathcal{G}}&=\frac{\alpha+1}{\kappa_2}g_*(q_{\mathcal{G}})2^{-(n-m)/2}.
\label{eq:CA_CG}
\end{align}
Then, in the large-$b$ and large-$K$ limits, we can analytically determine the reduction ratio as follows. 

\begin{figure}[!t]
\centering
\includegraphics[width=\linewidth]{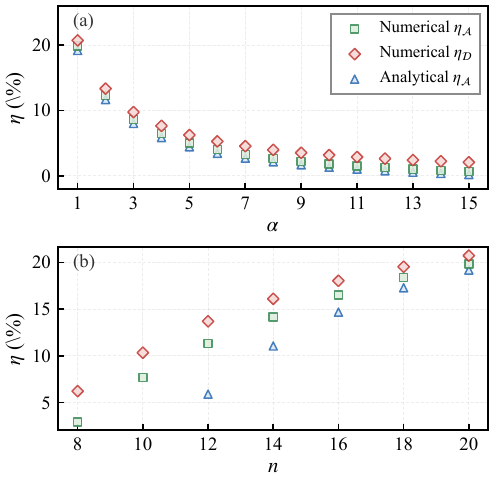}
\caption{Reduction ratios in the minimum expected circuit depth relative to the depth-optimized GRK baseline for $\ell_n=n$ and $\ell_m=m$. Panel~(a) shows the dependence on $\alpha=1,2,\ldots,15$ for $n=20$ and $m=10$, whereas panel~(b) shows the dependence on even $n=8,10,\ldots,20$ for $m=n/2$ and $\alpha=1$. Green squares and red diamonds denote the finite-size numerical values of $\eta_{\mathcal A}$ and $\eta_{\mathcal D}$, respectively; blue triangles denote the analytical $\eta_{\mathcal A}$ from Eq.~\eqref{eq:DR}. Every GRK baseline is obtained by direct integer depth optimization, and every numerical $S_{\mathcal A}$ point is recomputed from the exact finite-size evolution in the three-dimensional invariant subspace.}
\label{fig:emin_eta}
\end{figure}

\begin{corollary}\label{cor:DR}
In the limits $b\gg1$ and $K\gg1$, for $b^{-1}\ll q_{\mathcal{A}}\ll b^{1/2}$ and $q_{\mathcal{A}}>q_{\mathcal{A},\mathrm{c}}^{\mathrm{asy}}$, equivalently $\alpha<\alpha_{\mathrm{c},\mathcal{A}}^{\mathrm{asy}}$, we have
\begin{equation}
\eta_{\mathcal A} \simeq
\frac{n-m-n\bigl(C_{\mathcal{A}}+C_{\mathcal{G}}\bigr)}
{n\bigl(\alpha+1-C_{\mathcal{G}}\bigr)}\times 100\%.
\label{eq:DR}
\end{equation}
\end{corollary}

\begin{proof}
Dividing the approximate expected depths in Eqs.~\eqref{eq:GRKdepth_expand} and
\eqref{eq:Enew_expand} by $\kappa_0n2^{n/2}$ gives, to the same relative order,
\begin{subequations}
\begin{align}
\widehat E_{\mathcal{G}}=\frac{\mathrm{E}_{\ell,\mathcal{G}}^{\min}}{\kappa_0n2^{n/2}}
&\simeq(\alpha+1)-C_{\mathcal{G}},\\
\widehat E_{\mathcal{A}}=\frac{\mathrm{E}_{\ell,\mathcal{A}}^{\min,\mathrm{asy}}}{\kappa_0n2^{n/2}}
&\simeq\alpha+\frac mn+C_{\mathcal{A}}.
\label{eq:DR_truncated_depths}
\end{align}
\end{subequations}
Hence $\widehat E_{\mathcal{G}}-\widehat E_{\mathcal{A}}
=(n-m)/n-C_{\mathcal{A}}-C_{\mathcal{G}}$. Substituting this difference
into the definition of $\eta_{\mathcal A}$ and multiplying the numerator and denominator
by $n$ yields Eq.~\eqref{eq:DR}.
\end{proof}

Figure~\ref{fig:emin_eta} compares the analytical prediction of the reduction ratio in the minimum expected circuit depth in Eq.~\eqref{eq:DR} with finite-size numerical optimization. For the representative point $\alpha=1$, $n=20$, and $m=10$, Eq.~\eqref{eq:CA_CG} gives $C_{\mathcal A}\simeq0.0744$ and $C_{\mathcal G}\simeq0.0517$, so Eq.~\eqref{eq:DR} yields $\eta_{\mathcal A}\simeq19.19\%$. The numerical alternating result is $19.85\%$, while optimization over the full family $S_{\mathcal D}$ gives $\eta_{\mathcal D}=20.73\%$. Panel~(a) shows that the reduction decreases as global diffusion becomes relatively cheaper: over $\alpha=1,\ldots,15$, the numerical values of $\eta_{\mathcal A}$ and $\eta_{\mathcal D}$ fall from $19.85\%$ to $0.64\%$ and from $20.73\%$ to $2.06\%$, respectively. The analytical $\eta_{\mathcal A}$ follows the same trend. Panel~(b) fixes $\alpha=1$ and $m=n/2$: as $n$ increases from $8$ to $20$, the numerical $\eta_{\mathcal A}$ and $\eta_{\mathcal D}$ rise from $2.93\%$ to $19.85\%$ and from $6.22\%$ to $20.73\%$, respectively. In the limit $n\to\infty$ along $m=n/2$, the correction terms $C_{\mathcal A}$ and $C_{\mathcal G}$ vanish and Eq.~\eqref{eq:DR} gives $\eta_{\mathcal A}\to25\%$. This limit has a direct leading-order interpretation: for $\alpha=1$, the normalized leading depths above satisfy $\widehat E_{\mathcal G}\to2$ and $\widehat E_{\mathcal A}\to3/2$, so the alternating sequence removes one quarter of the leading GRK expected depth. At finite $n$, $C_{\mathcal A}>0$ raises the alternating depth while $C_{\mathcal G}>0$ lowers the GRK baseline; both effects reduce the observed gain below $25\%$, explaining the upward trend in panel~(b) and the approximately $20\%$ reduction at $n=20$.

\section{Conclusions}
\label{sec:conclusions}

In this work, we seek a quantum partial-search algorithm optimized for circuit depth. We propose the alternating sequence $S_{\mathcal{A}}=(G_nG_m^{\tilde k})^k$ and optimize its expected circuit depth. The $O(3)$ geometric analysis yields the exact success probability in Lemma~\ref{lem:exactProb} and its asymptotic form in Corollary~\ref{lem:asympProb}. The analytical optimal parameters and minimum expected depths of $S_\mathcal{A}$ in the three asymptotic regimes are given in Theorems~\ref{thm:MED_caseI}, \ref{thm:MED_caseII}, and~\ref{thm:MED_caseIII} and Corollaries~\ref{cor:MED_caseI_depth} and~\ref{cor:MED_caseII_depth}. Comparison with the depth-optimized GRK sequence gives the critical depth ratio $\alpha_{\mathrm{c},\mathcal{A}}$, below which $S_{\mathcal A}$ is shallower than GRK (Theorem~\ref{thm:alpha_c}). Corollary~\ref{cor:alpha_c_asymp} gives the asymptotic critical depth ratios, while Corollary~\ref{cor:alpha_c_max} gives the leading-order maximum $\alpha_{\mathrm{c},\mathcal{A}}^{\max}\simeq0.2976\times2^{n/3}$. At $\alpha=1$, $n=20$, and $m=10$, the analytical and numerical reductions in the minimum expected circuit depth are both approximately $20\%$. For $m=n/2$, we find that the reduction in the minimum expected circuit depth reaches $\eta_{\mathcal A}\to25\%$ as $n\to\infty$. Our work systematically develops a depth-efficient partial-search algorithm and determines the critical boundary below which it outperforms the conventional GRK algorithm in minimum expected circuit depth.

A direct extension of our work is to develop depth-efficient partial-search algorithms for multiple targets. The GRK algorithm has been generalized to multiple-target partial search~\cite{Choi2007QuantumPS,zhong2009quantum,Zhang2017QuantumPS}, but other partial-search operator sequences may be more suitable than the alternating sequence $S_{\mathcal A}$ in the multiple-target setting. The second question is how to design a depth-efficient partial-search algorithm valid for $m>\lfloor n/2\rfloor$, beyond the single-iteration case established in Theorem~\ref{thm:MED_caseIII}. To our knowledge, neither a depth-efficient nor an oracle-efficient partial-search algorithm has been established for $m>\lfloor n/2\rfloor$. The third extension of our work is to design a depth-efficient algorithm for partial amplitude amplification when the initial state has nonuniform amplitudes, building on the general amplitude-amplification framework and analyses of search from arbitrary initial amplitude distributions~\cite{brassard2000quantum,biham1999arbitrary}. We leave these questions for future study.

\begin{acknowledgments}
This work was supported by the National Natural Science Foundation of China under Grant Nos. 12305028, 12275215, and 12247103 and by the Youth Innovation Team of Shaanxi Universities. KZ was supported by the China Postdoctoral Science Foundation under Grant No. 2025M773421, the Shaanxi Province Postdoctoral Science Foundation under Grant No. 2025BSHYDZZ017, and the Scientific Research Program funded by the Education Department of the Shaanxi Provincial Government under Program No. 24JP186. VK was supported by the U.S. Department of Energy, Office of Science, National Quantum Information Science Research Centers, Co-Design Center for Quantum Advantage (C2QA), under Contract No. DE-SC0012704.
\end{acknowledgments}

\appendix

\section{Numerical results for depth-optimal partial-search sequences}
\label{app:tables}

The depth-optimal partial-search sequences for $\alpha=1$ are presented in Table~\ref{tab:n10_alpha1_unrestricted}. Table~\ref{tab:n10_reqGn_combined} lists the supplementary results for $\alpha=3,5,10,100$. The depth-optimal sequences gradually approach the GRK form as $\alpha$ increases. At $\alpha=100$, all the depth-optimal sequences have the GRK form.

\begin{table*}[t]
\centering
\footnotesize
\caption{Depth-optimal partial-search sequences and the minimum expected depths for $\alpha=3,5,10,100$, with $n=10$, $\ell_n=n$, and $\ell_m=m$. We set the enumeration cutoff to $k_{\mathrm{tot}}\le24$ based on the bound in Eq.~\eqref{eq:enumeration_lower_bound}.}
\label{tab:n10_reqGn_combined}
\setlength{\tabcolsep}{4pt}
\begin{tabular}{@{}clcccc@{}}
\toprule
\multicolumn{6}{c}{$\alpha=3$} \\
\midrule
$m$ & Optimal sequence $S$ & $k_{\mathrm{tot}}$ & $\ell(S)$ & $\pr(S)$ & $\mathrm{E}_{\ell}^{\min}$ \\
\midrule
2 & $G_n G_m G_n^{16}$ & 18 & 712 & 0.8386 & 849.08 \\
3 & $G_n G_m^2 (G_n G_m)^8$ & 19 & 690 & 0.8482 & 813.46 \\
4 & $G_n G_m^2 (G_n G_m)^7$ & 17 & 626 & 0.8104 & 772.50 \\
5 & $G_n G_m^3 (G_n G_m)^6$ & 16 & 595 & 0.8021 & 741.79 \\
6 & $G_n$ & 1 & 40 & 0.0698 & 572.98 \\
7 & $G_n$ & 1 & 40 & 0.1318 & 303.44 \\
8 & $G_n$ & 1 & 40 & 0.2558 & 156.34 \\
9 & $G_n$ & 1 & 40 & 0.5039 & 79.38 \\
\bottomrule
\end{tabular}
\qquad
\begin{tabular}{@{}clcccc@{}}
\toprule
\multicolumn{6}{c}{$\alpha=5$} \\
\midrule
$m$ & Optimal sequence $S$ & $k_{\mathrm{tot}}$ & $\ell(S)$ & $\pr(S)$ & $\mathrm{E}_{\ell}^{\min}$ \\
\midrule
2 & $G_n G_m G_n^{16}$ & 18 & 1072 & 0.8386 & 1278.39 \\
3 & $G_n G_m G_n^{15}$ & 17 & 1013 & 0.8142 & 1244.12 \\
4 & $G_n G_m^2 (G_n G_m)^7$ & 17 & 966 & 0.8104 & 1192.07 \\
5 & $G_n G_m^3 (G_n G_m)^6$ & 16 & 915 & 0.8021 & 1140.73 \\
6 & $G_n$ & 1 & 60 & 0.0698 & 859.48 \\
7 & $G_n$ & 1 & 60 & 0.1318 & 455.16 \\
8 & $G_n$ & 1 & 60 & 0.2558 & 234.51 \\
9 & $G_n$ & 1 & 60 & 0.5039 & 119.07 \\
\bottomrule
\end{tabular}

\vspace{8pt}

\begin{tabular}{@{}clcccc@{}}
\toprule
\multicolumn{6}{c}{$\alpha=10$} \\
\midrule
$m$ & Optimal sequence $S$ & $k_{\mathrm{tot}}$ & $\ell(S)$ & $\pr(S)$ & $\mathrm{E}_{\ell}^{\min}$ \\
\midrule
2 & $G_n G_m G_n^{16}$ & 18 & 1972 & 0.8386 & 2351.66 \\
3 & $G_n G_m G_n^{15}$ & 17 & 1863 & 0.8142 & 2288.06 \\
4 & $G_n G_m^2 G_n^{14}$ & 17 & 1858 & 0.8299 & 2238.76 \\
5 & $G_n G_m^3 (G_n G_m)^6$ & 16 & 1715 & 0.8021 & 2138.10 \\
6 & $G_n$ & 1 & 110 & 0.0698 & 1575.71 \\
7 & $G_n$ & 1 & 110 & 0.1318 & 834.45 \\
8 & $G_n$ & 1 & 110 & 0.2558 & 429.94 \\
9 & $G_n$ & 1 & 110 & 0.5039 & 218.30 \\
\bottomrule
\end{tabular}
\qquad
\begin{tabular}{@{}clcccc@{}}
\toprule
\multicolumn{6}{c}{$\alpha=100$} \\
\midrule
$m$ & Optimal sequence $S$ & $k_{\mathrm{tot}}$ & $\ell(S)$ & $\pr(S)$ & $\mathrm{E}_{\ell}^{\min}$ \\
\midrule
2 & $G_n G_m G_n^{16}$ & 18 & 18172 & 0.8386 & 21670.54 \\
3 & $G_n G_m G_n^{15}$ & 17 & 17163 & 0.8142 & 21078.86 \\
4 & $G_n G_m^2 G_n^{14}$ & 17 & 17158 & 0.8299 & 20674.15 \\
5 & $G_n G_m^3 G_n^{12}$ & 16 & 16145 & 0.8089 & 19958.52 \\
6 & $G_n$ & 1 & 1010 & 0.0698 & 14467.86 \\
7 & $G_n$ & 1 & 1010 & 0.1318 & 7661.81 \\
8 & $G_n$ & 1 & 1010 & 0.2558 & 3947.66 \\
9 & $G_n$ & 1 & 1010 & 0.5039 & 2004.37 \\
\bottomrule
\end{tabular}
\end{table*}

\section{Algebraic structure and success probability of the alternating sequence}
\label{app:recurrence}

This appendix derives the exact success probability in Lemma~\ref{lem:exactProb}. The derivation is based on two ingredients: the $O(3)$ spectral decomposition of $G_nG_m^{\tilde k}$ and the amplitude recurrence on its two-dimensional invariant subspace.

\subsection{\texorpdfstring{$O(3)$}{O(3)} representation and rotation angle}

The operators $G_n$ and $G_m^{\tilde{k}}$ are orthogonal matrices, with $\det G_n=-1$ and $\det G_m^{\tilde{k}}=1$. To shorten the notation, define $T_{\tilde{k}}=G_nG_m^{\tilde{k}}$. The operator $T_{\tilde{k}}$ lies in $O(3)$ and has determinant $-1$. Its spectrum consists of one eigenvalue $-1$ and a conjugate pair $\mathrm{e}^{\pm\mathrm{i}\Omega}$. The space therefore decomposes into the $(-1)$ eigenspace and a two-dimensional invariant subspace $\mathcal{V}_2$, on which $T_{\tilde{k}}$ acts as a rotation by $\Omega$.

The normalized eigenvector with eigenvalue $-1$ is
\begin{equation}
|v_-\rangle
=\frac{1}{\sqrt{1-\sin^2\gamma\sin^2\lambda}}
\begin{pmatrix}
\cos\gamma\sin(\tilde{k}\theta_2)\\
-\cos\gamma\cos(\tilde{k}\theta_2)\\
\sin\gamma\cos\lambda
\end{pmatrix}.
\label{eq:v_minus}
\end{equation}
This eigenvector satisfies $\langle v_-|s_n\rangle=0$, where the initial state $|s_n\rangle$ is given by Eq.~\eqref{eq:sn_three}. The initial state therefore lies entirely in $\mathcal{V}_2$. The rotation angle follows from the trace. Let $\phi=2\tilde{k}\theta_2$. Using Eq.~\eqref{eq:matrix_GnGm} and $2\theta_2+\phi=2\lambda$, one obtains
\begin{equation}
\tr(T_{\tilde{k}}) = 1 - 4\sin^2\gamma\sin^2\lambda .
\end{equation}
Since $\tr(T_{\tilde{k}})=-1+2\cos\Omega$, we write the rotation angle as
\begin{equation}
\label{eq:Omega_half_app}
\Omega=2\arcsin(\sin\gamma\sin\lambda),
\end{equation}
which is in the range $\Omega\in[-\pi,\pi]$. Note that $\sin(\Omega/2)=\sin\gamma\,\sin\lambda$.

\subsection{Amplitude recurrence}

On $\mathcal{V}_2$, the restriction of $T_{\tilde{k}}$ is a planar rotation. Therefore it satisfies the quadratic identity
\begin{equation}
T_{\tilde{k}}^2 - 2\cos\Omega\,T_{\tilde{k}} + I = 0,
\label{eq:CH_quadratic}
\end{equation}
when acting on vectors in $\mathcal{V}_2$. Equivalently, let $\mu$ denote an eigenvalue of $T_{\tilde{k}}$. The full characteristic equation is
\begin{equation}
\det(\mu I-T_{\tilde{k}})
=(\mu+1)(\mu^2 - 2\cos\Omega\,\mu + 1)=0,
\label{eq:char_poly}
\end{equation}
but the $-1$ eigenspace component is absent from the evolution of $|s_n\rangle$. Denote $a_j = \langle \bar{b} | T_{\tilde{k}}^{\,j} | s_n \rangle$. Applying $\langle\bar b|$ to Eq.~\eqref{eq:CH_quadratic} acting on $T_{\tilde{k}}^j|s_n\rangle$ yields the second-order linear recurrence relation
\begin{equation}
a_{j+2} - 2\cos\Omega\,a_{j+1} + a_j = 0,
\end{equation}
with characteristic roots $\mathrm{e}^{\pm \mathrm{i}\Omega}$.

The initial conditions are
\begin{equation}
a_0 = \cos\gamma, \quad a_1 = \cos\gamma\,(2\cos\Omega-1).
\end{equation}
Writing the general solution in trigonometric form $a_j=A\cos(j\Omega)+B\sin(j\Omega)$, the initial conditions give
\begin{equation}
A = \cos\gamma,\quad B = -\cos\gamma\,\tan\frac{\Omega}{2}.
\end{equation}
Thus we obtain 
\begin{multline}
a_j = \cos\gamma\left(\cos(j\Omega)-\tan\frac{\Omega}{2}\sin(j\Omega)\right)\\
= \frac{\cos\gamma}{\cos(\Omega/2)}\cos\!\left(\frac{(2j+1)\Omega}{2}\right),
\label{eq:a_j_final}
\end{multline}
The success probability is $1-|a_k|^2$. Substituting Eq.~\eqref{eq:Omega_half_app} into Eq.~\eqref{eq:a_j_final} gives Eq.~\eqref{eq:exactP}.

\section{Asymptotic optimization of expected circuit depth}
\label{app:proof_optimal}

This appendix proves Theorems~\ref{thm:MED_caseI}, \ref{thm:MED_caseII}, and~\ref{thm:MED_caseIII}, together with Corollaries~\ref{cor:MED_caseI_depth} and~\ref{cor:MED_caseII_depth} in Sec.~\ref{sec:analytical}. The optimization of expected depth in Eq.~\eqref{eq:E_xy} is separable: for fixed $x$, the $y$ dependence is $y/\sin^2 y$, while all $x$-dependent factors are grouped into $h(x)$ defined in Eq.~\eqref{eq:h(x)}. This separability allows a two-step optimization procedure over the full domain where the success probability is nonzero.

\subsection{Continuous optimization}

We first treat $y$ as a continuous variable. For $K\gg 1$, the optimal $k$ scales as $\Theta(\sqrt{K}/|\sin(x\theta_2)|)$, consistent with Corollary~\ref{lem:asympProb}, so the continuous approximation is accurate at leading order. Taking the derivative of $y/\sin^2 y$ with respect to $y$ and setting it to zero gives the stationary condition $\tan y = 2y$. The smallest positive root of this transcendental equation is denoted $y_0 \approx 1.1656$, which coincides with the optimal stopping point for Grover's algorithm given by Eq.~\eqref{eq:E_min_Grover} in the large-database limit.

Solving the definition of $y$ in Eq.~\eqref{eq:intro_y} for the corresponding optimal $k$, we obtain
\begin{equation}
k_*(x) = \frac{y_0\sqrt{K}}{2\sin(x\theta_2)} - \frac{1}{2} + \mathcal{O}(1).
\label{eq:k_opt_x_app}
\end{equation}
For sufficiently large $K$ and $k\ge 1$, the bounded correction is negligible relative to the leading term. Its relative size is at most $\mathcal{O}(K^{-1/2})$.

\subsection{Cost-effective local-iteration regime (Case I)}

The optimal value of the inner iteration count $\tilde k$ is determined by the function $h(x)$ given by Eq.~\eqref{eq:h(x)} with $x=\tilde k+1$. For any admissible $x>0$ outside the first positive quarter-period, we reduce $x\theta_2$ modulo $\pi$ and, if necessary, reflect the remainder about $\pi/2$. The resulting angle lies in $(0,\pi/2]$ and has the same absolute sine as $x\theta_2$, while the corresponding value of $x$ is smaller. Because the numerator is strictly increasing in $x$, this folded point gives a smaller value of $h(x)$. All later half-periods and the second quadrant of the first period are therefore excluded from the continuous global minimum. On the remaining interval $0<x\le\pi/(2\theta_2)$, $\sin(x\theta_2)$ is positive.

Treating $x$ as a positive real variable on this interval and differentiating $h(x)$ gives the stationary condition
\begin{equation}
\tan(x\theta_2) = \theta_2\left(x + \frac{\ell_\Delta}{\ell_{G_m}}\right).
\end{equation}
Introducing the variable $z=x\theta_2$ and the alternating-family depth ratio $q_{\mathcal{A}}=\ell_\Delta/\ell_{G_m}$, the equation simplifies to
\begin{equation}
\tan z = z + q_{\mathcal{A}}\theta_2.
\label{eq:stationary_z_app}
\end{equation}
For the interior solution in the large-block limit $b\gg 1$, we have $z\ll 1$, so we may expand $\tan z$ in a Taylor series
\begin{equation}
\tan z = z + \frac{z^3}{3} + \frac{2z^5}{15} + \mathcal{O}(z^7).
\end{equation}
Substituting into Eq.~\eqref{eq:stationary_z_app} yields
\begin{equation}
\frac{z^3}{3} + \frac{2z^5}{15} + \mathcal{O}(z^7) = q_{\mathcal{A}}\theta_2.
\label{eq:z_expansion_app}
\end{equation}
We solve this equation perturbatively order by order. Denote the leading-order approximation to $z$ by $z_{\mathrm{LO}}$. The leading-order solution satisfies $z_{\mathrm{LO}}^3/3=q_{\mathcal{A}}\theta_2$, giving
\begin{equation}
z_{\mathrm{LO}}=(3q_{\mathcal{A}}\theta_2)^{1/3}.
\end{equation}
To obtain the subleading correction, set $z=z_{\mathrm{LO}}+\Delta z$ and substitute into Eq.~\eqref{eq:z_expansion_app}. Matching the next order gives $\Delta z=-2z_{\mathrm{LO}}^3/15+\mathcal{O}(z_{\mathrm{LO}}^5)$. The stationary point through subleading order is therefore
\begin{equation}
z_*=(3q_{\mathcal{A}}\theta_2)^{1/3}-\frac{2}{5}q_{\mathcal{A}}\theta_2
+\mathcal{O}\left((q_{\mathcal{A}}\theta_2)^{5/3}\right).
\label{eq:z_star_app}
\end{equation}
Transforming back to $x$ via $x=z/\theta_2$, we obtain the continuous optimum
\begin{equation}
x_*=(3q_{\mathcal{A}})^{1/3}\theta_2^{-2/3}-\frac{2}{5}q_{\mathcal{A}}
+\mathcal{O}\left(q_{\mathcal{A}}^{5/3}\theta_2^{2/3}\right).
\label{eq:x_star_app}
\end{equation}
The expansion presumes $z_*\ll1$ while $x_*$ is asymptotically separated
from the first few integers. Since $\theta_2\sim b^{-1/2}$, these conditions give the interior asymptotic window $b^{-1}\ll q_{\mathcal{A}}\ll b^{1/2}$.

The stationary point $x_*$ is a strict local minimum of $h(x)$. From
\begin{equation}
h'(x) = \frac{\ell_{G_m}\sin z - \theta_2(\ell_{G_m}x+\ell_\Delta)\cos z}{\sin^2 z},
\end{equation}
the stationary condition $\tan z_* = z_* + q_{\mathcal{A}}\theta_2$ yields
\begin{equation}
h''(x_*) = \frac{\ell_{G_m}\theta_2}{\cos z_*} > 0,
\label{eq:h_double_prime}
\end{equation}
which directly confirms that $x_*$ is a local minimum. Since $h(x)$ has only this single stationary point in $(0,\pi/(2\theta_2))$, it is the minimum on this interval. The preceding folding argument makes it the continuous global minimum over all $x>0$ with nonzero success probability.

Substituting $z_*$ into Eq.~\eqref{eq:k_opt_x_app} gives the optimal $k$. Expanding $\sin z_*$ to subleading order,
\begin{equation}
\sin z_* = z_* - \frac{z_*^3}{6} + \mathcal{O}(z_*^5) = z_{\mathrm{LO}} - \frac{3}{10}z_{\mathrm{LO}}^3 + \mathcal{O}(z_{\mathrm{LO}}^5),
\end{equation}
and therefore
\begin{equation}
\frac{1}{\sin z_*} = \frac{1}{z_{\mathrm{LO}}}\left(1 + \frac{3}{10}z_{\mathrm{LO}}^2 + \mathcal{O}(z_{\mathrm{LO}}^4)\right).
\label{eq:inv_sin_series}
\end{equation}
Substituting Eq.~\eqref{eq:inv_sin_series} into Eq.~\eqref{eq:k_opt_x_app}, with $z_{\mathrm{LO}}=(3q_{\mathcal{A}}\theta_2)^{1/3}$, we arrive at
\begin{multline}
k_*=\frac{y_0\sqrt K}{2(3q_{\mathcal{A}}\theta_2)^{1/3}}\\
\times\left(1+\frac{3^{5/3}}{10}(q_{\mathcal{A}}\theta_2)^{2/3}
+\mathcal{O}\bigl((q_{\mathcal{A}}\theta_2)^{4/3}\bigr)\right)
-\frac{1}{2}+\mathcal{O}(1).
\label{eq:k_int_app}
\end{multline}
The $\mathcal{O}(1)$ term collects the bounded effects of integer rounding and finite-$K$ corrections to the continuous outer stationary point in Eq.~\eqref{eq:k_opt_x_app}. It does not modify the displayed leading or scaled subleading terms.

Finally, we substitute $y=y_0$ and the stationary relation $\ell_{G_m}x_*+\ell_\Delta = \ell_{G_m}\tan z_*/\theta_2$ into the leading term of Eq.~\eqref{eq:E_xy}, which gives the minimal expected depth
\begin{multline}
\mathrm{E}_{\ell,\mathcal{A}}^{\min,\mathrm{I}}
= \frac{\sqrt{K}}{2}\frac{y_0}{\sin^2 y_0}
\frac{\ell_{G_m}\tan z_*}{\theta_2\sin z_*}
\Bigl(1+\mathcal{O}(K^{-1/2})\Bigr)\\
= \kappa_0\,\ell_{G_m}\,2^{n/2}\frac{1}{\cos z_*}
\Bigl(1+\mathcal{O}(K^{-1/2})+\mathcal{O}(b^{-1})\Bigr),
\end{multline}
where we have used $\sqrt K/\theta_2=2^{n/2}\bigl(1+\mathcal{O}(b^{-1})\bigr)$ and defined $\kappa_0 = y_0/(2\sin^2 y_0) \approx 0.6900$. For $q_{\mathcal{A}}=\mathcal{O}(1)$, $z_*=\mathcal{O}(2^{-m/6})$, and the relative correction introduced by $1/\cos z_*$ is $\mathcal{O}(2^{-m/3})$. If $q_{\mathcal{A}}=\Theta(1)$, then $z_*=\Theta(2^{-m/6})$ and this correction is $\Theta(2^{-m/3})$. Meanwhile,
\begin{equation}
\frac{K^{-1/2}}{2^{-m/3}}=2^{-(3n-5m)/6}, \qquad \frac{b^{-1}}{2^{-m/3}}=2^{-2m/3},
\end{equation}
both vanish for $m<n/2+\mathcal{O}(1)$ as $b,K\to\infty$, with $q_{\mathcal{A}}=\mathcal{O}(1)$. Thus the combined remainder is $o(2^{-m/3})$, which completes the proofs of Theorem~\ref{thm:MED_caseI} and Corollary~\ref{cor:MED_caseI_depth}.

\subsection{Boundary solution (Case II)}

The preceding asymptotic optimization, which treats $x$ as a continuous variable, cannot determine the first integer switch, since $x=\tilde k+1$ is a positive integer. The continuous stationary point equals $2$ at
\begin{equation}
q_{\mathrm{stat}}(b)=\frac{\tan(2\theta_2)}{\theta_2}-2
=\frac{8}{3b}+\mathcal{O}(b^{-2}).
\end{equation}
However, the exact integer transition between $x=1$ and $x=2$ occurs at a smaller value of $q_{\mathcal{A}}$. We set $q_{\mathcal{A}}=\ell_\Delta/\ell_{G_m}$ and recall that $\theta_2=\arcsin(b^{-1/2})$. The comparison between the first two admissible integers, $h(2)>h(1)$, gives $q_{\mathcal{A}}<q_{\mathrm{bd}}$, where
\begin{equation}
\label{eq:qbd_exact_app}
q_{\mathrm{bd}} = \frac{2(1-\cos\theta_2)}{2\cos\theta_2-1} = b^{-1}+\mathcal{O}(b^{-2}).
\end{equation}
Consequently, $q_{\mathrm{bd}}<q_{\mathrm{stat}}$ for large $b$. For
$q_{\mathcal{A}}=\mathcal{O}(b^{-1})$ with $q_{\mathcal{A}}\ge q_{\mathrm{bd}}$, the correct reduced solution
is the exact integer minimizer $x_{\mathrm{int}}$ in Eq.~\eqref{eq:x_integer}, not
the rounded interior asymptotic series. Substitution into Eq.~\eqref{eq:E_xy} gives
Eq.~\eqref{eq:MED_boundary_layer} directly.

The exact integer search in Eq.~\eqref{eq:x_integer} is finite. The denominator of $h(x)$ never exceeds unity, whereas its numerator increases linearly with $x$. Once this numerator reaches or exceeds the candidate value $h(1)$, no larger integer can improve the minimum. Thus only finitely many integer values of $x$ need to be examined.

For $q_{\mathcal{A}}<q_{\mathrm{bd}}$, the unique first-quadrant stationary point lies below $x=2$. This follows from $q_{\mathrm{bd}}<q_{\mathrm{stat}}$, because the stationary point moves to larger $x$ as $q_{\mathcal{A}}$ increases. Hence $h(x)$ is strictly increasing from $x=2$ to the end of the first quadrant. Together with $h(2)>h(1)$ from Eq.~\eqref{eq:qbd_exact_app}, this excludes every integer $x\ge2$ in that interval.

It remains to exclude the rest of the global domain. For $x\ge\pi/(2\theta_2)$, we have
\[
h(x)\ge \ell_{G_m}x+\ell_\Delta
\ge \frac{\pi}{2\theta_2}\ell_{G_m}+\ell_\Delta,
\]
whereas $h(1)=\ell_{G_m}(1+q_{\mathcal{A}})/\sin\theta_2=\ell_{G_m}\sqrt b\bigl(1+\mathcal{O}(b^{-1})\bigr)$ because $q_{\mathcal{A}}<q_{\mathrm{bd}}=\mathcal{O}(b^{-1})$. Since $\pi/2>1$, all $x$ in later quadrants and periods satisfy $h(x)>h(1)$ for sufficiently large $b$. Thus $x=1$ is the global integer minimum, not merely a local or fixed-$x$ asymptotic minimum.

When $q_{\mathcal{A}}<q_{\mathrm{bd}}$ (or $\alpha>\alpha_{\mathrm{bd}}$ under the linear-depth specialization), the optimum lies at $\tilde{k}=0$. In this regime, the expected depth as a function of $k$ takes the form
\begin{equation}
\mathrm{E}_{\ell,\mathcal{A}}(k,0)=\frac{k(\alpha+1)\ell_n}{\sin^2\left((2k+1)/2^{n/2}\right)}\bigl(1+\mathcal{O}(K^{-1})\bigr).
\end{equation}
Defining $\zeta = (2k+1)/2^{n/2}$, we have $\mathrm{E}_{\ell,\mathcal{A}}(k,0) \propto \zeta/\sin^2\zeta$, whose minimum is located at $\zeta = y_0$. The optimal iteration count and minimum expected depth are therefore
\begin{equation}
k_*=\frac{y_0}{2}2^{n/2}-\frac12+\mathcal{O}(1),
\end{equation}
and
\begin{equation}
\mathrm{E}_{\ell,\mathcal{A}}^{\min,\mathrm{II}} = \kappa_0(\alpha+1)\ell_n\,2^{n/2}\bigl(1+\mathcal{O}(2^{-n/2})\bigr),
\end{equation}
which completes the proofs of Theorem~\ref{thm:MED_caseII} and Corollary~\ref{cor:MED_caseII_depth}.

\subsection{Large-\texorpdfstring{$m$}{m} regime (Case III)}

The preceding cases assume that $y=(2k+1)\sin(x\theta_2)/\sqrt K$ can reach $y_0$, with $k=\mathcal{O}(\sqrt K/\sin\lambda)$ as specified in Corollary~\ref{lem:asympProb}. To locate the loss of this asymptotic advantage, set $m=n/2+c$ with $c=\mathcal{O}(1)$. The success probability of a single application of $G_n$ is $K^{-1}\bigl(1+o(1)\bigr)$, so
\begin{equation}
\mathrm{E}_{\ell}(G_n)\sim n(\alpha+1)2^{n/2-c}.
\end{equation}
For the competing Case~I solution, $m/n=1/2+o(1)$ and Eq.~\eqref{eq:Enew_expand} gives
\begin{equation}
\mathrm{E}_{\ell,\mathcal{A}}^{\min,\mathrm{asy}}
\sim\kappa_0n(\alpha+1/2)2^{n/2}.
\end{equation}
Equating the leading terms yields $c_*(\alpha)$ in
Eq.~\eqref{eq:caseIII_crossover}. For $c>c_*(\alpha)+o(1)$, a single global
iteration has the smaller expected depth. This separates the large-$m$ crossover from
the $x=1$-to-$x=2$ boundary layer. Subleading corrections determine
the nearest finite-$n$ integer threshold. The loss of partial-search advantage
is consistent with the query-complexity analysis of
Ref.~\cite{jiang2026exactboundsquantumpartial}.

In this regime, the optimal sequence reduces to a single global operator, $S_{\mathcal{A}}(1,0)=G_n$, with circuit depth $\ell(G_n)=(\alpha+1)\ell_n$. In the expression from Lemma~\ref{lem:exactProb}, setting $k=1$ and $\tilde{k}=0$ (so that $\lambda=\theta_2$ and $\sin\lambda=1/\sqrt{b}$) gives the success probability
\begin{equation}
\pr_{\mathcal{A}}(1,0)=1-\frac{K-1}{K-1/b}\cos^2\!\left(3\arcsin\frac{1}{\sqrt{N}}\right).
\end{equation}
Expanding in powers of $N^{-1}$, we obtain
\begin{equation}
\pr_{\mathcal{A}}(1,0) = \frac{1}{K}\left(1+\frac{8(K-1)}{Kb}\right)+\mathcal{O}(N^{-2}).
\end{equation}
Factoring out $K^{-1}$ without imposing $K\gg1$ gives
\begin{equation}
\pr_{\mathcal{A}}(1,0)
=\frac{1}{K}\left(1+\frac{8(K-1)}{N}
+\mathcal{O}\!\left(\frac{K}{N^2}\right)\right).
\label{eq:caseIII_finiteK_prob}
\end{equation}
The absolute $\mathcal{O}(N^{-2})$ remainder becomes
$\mathcal{O}(K/N^2)$ inside the parentheses. This form is uniform when
$K=\mathcal{O}(1)$, which includes $m=n-\mathcal{O}(1)$. Replacing
$(K-1)/K$ by unity would not be uniform there. The minimum expected circuit
depth is therefore
\begin{multline}
\mathrm{E}_{\ell,\mathcal{A}}^{\min,\mathrm{III}}
=\frac{(\alpha+1)\ell_n}{\pr_{\mathcal{A}}(1,0)}
\\
=(\alpha+1)\ell_n\frac{K}{1+8(K-1)/N}
\left(1+\mathcal{O}\!\left(\frac{K}{N^2}\right)\right),
\end{multline}
which completes the proof of Theorem~\ref{thm:MED_caseIII}.

\section{GRK depth optimization and critical-depth-ratio analysis}
\label{app:critical}

This appendix supplies the derivations supporting the GRK comparison and critical depth ratios in Sec.~\ref{sec:critical}. It first proves Lemma~\ref{lem:GRK_opt_parameters} and Theorem~\ref{thm:GRKdepth} by optimizing the total and local GRK iterations through order $\sqrt b$. It then proves Theorem~\ref{thm:alpha_c} by obtaining the reduced-integer crossing equation and establishing the uniqueness of its physical root. Finally, it derives the interior continuous crossover equation and proves Corollaries~\ref{cor:alpha_c_asymp} and~\ref{cor:alpha_c_max}, including the small- and large-$m$ asymptotic roots, their validity regimes, and the location and scaling of the maximum critical depth ratio.

\subsection{Depth optimization of the GRK algorithm}
\label{app:depth_GRK}

\begin{proof}[Proof of Lemma~\ref{lem:GRK_opt_parameters}]
Using $\sqrt b=\sqrt N/\sqrt K$ and
$n-m=(\alpha+1)nq_{\mathcal{G}}$, Eqs.~\eqref{eq:GRK_prob_depthopt}
and~\eqref{eq:GRK_expected_depth} give, through order $\sqrt b$,
\begin{multline}
\mathrm{E}_{\ell,\mathcal{G}}(v_{\mathrm{tot}},v_2)
=\frac{(\alpha+1)n v_{\mathrm{tot}}\sqrt N}
{\sin^2(2v_{\mathrm{tot}})}\\
-\frac{(\alpha+1)n\sqrt b}{\sin^2(2v_{\mathrm{tot}})}
\left(q_{\mathcal{G}}v_2
+\frac{2v_{\mathrm{tot}}(\sin(2v_2)-v_2)\sin(4v_{\mathrm{tot}})}
{\sin^2(2v_{\mathrm{tot}})}\right)\\
+o((\alpha+1)n\sqrt b).
\label{eq:GRK_expected_depth_expand}
\end{multline}
Here $v_{\mathrm{tot}}$ and $v_2$ are the total and local GRK iteration
counts rescaled by $\sqrt N$ and $\sqrt b$, respectively, as defined in
Eq.~\eqref{eq:GRK_rescaled_iterations}. The parameter $q_{\mathcal G}$,
defined in Eq.~\eqref{eq:qG_definition}, is the fractional depth saved by
replacing a global Grover iteration with a local one. The numerator and the
probability correction in Eq.~\eqref{eq:GRK_expected_depth_expand} follow from
Eqs.~\eqref{eq:GRK_single_depth} and~\eqref{eq:GRK_prob_depthopt}, respectively.
For fixed $v_2$, the leading coefficient is
$v_{\mathrm{tot}}/\sin^2(2v_{\mathrm{tot}})$. Stationarity in the first
physical interval gives $\tan(2v_{\mathrm{tot}})=4v_{\mathrm{tot}}$, hence
$v_{\mathrm{tot}}=\kappa_2=y_0/2$. Evaluating the order-$\sqrt b$ term there
and using $2\kappa_2\sin(2y_0)/\sin^2y_0=1$ yield
\begin{multline}
\min_{v_{\mathrm{tot}}}
\mathrm{E}_{\ell,\mathcal{G}}(v_{\mathrm{tot}},v_2)
=\kappa_0(\alpha+1)n\sqrt N\\
-\frac{(\alpha+1)n}{\sin^2y_0}
\left(\sin(2v_2)-(1-q_{\mathcal{G}})v_2\right)\sqrt b\\
+o((\alpha+1)n\sqrt b).
\label{eq:GRKdepth_v2_expand}
\end{multline}
Thus one maximizes $g(v_2)=\sin(2v_2)-(1-q_{\mathcal G})v_2$.
Because $g''(v_2)=-4\sin(2v_2)<0$ on $0<v_2<\pi/2$, its stationary
condition gives the unique maximizer in Eq.~\eqref{eq:GRK_v2star}, with
$\pi/6\le v_{2,*}\le\pi/4$. The strict leading-order minimum implies that the
order-$K^{-1/2}$ term shifts $v_{\mathrm{tot},*}$ only by
$\mathcal{O}(K^{-1/2})$, completing the proof.
\end{proof}

\begin{proof}[Proof of Theorem~\ref{thm:GRKdepth}]
Substituting $v_{2,*}$ into Eq.~\eqref{eq:GRKdepth_v2_expand} gives
Eq.~\eqref{eq:GRKdepth_expand}. The probability remainders
$\mathcal{O}(K^{-1})$ and $\mathcal{O}(N^{-1/2})$, the quadratic effect of the
$\mathcal{O}(K^{-1/2})$ shift in $v_{\mathrm{tot},*}$, and integer rounding all
contribute $o((\alpha+1)n\sqrt b)$. They therefore do not change the displayed
order-$\sqrt b$ coefficient.
\end{proof}

Finally, the envelope theorem gives $g_*'(q_{\mathcal G})=v_{2,*}$. Since
$v_{2,*}=\pi/6$ at $q_{\mathcal G}=0$, expanding $g_*$ and using
$g_*(0)/\sin^2y_0=\kappa_4$ and
$(\pi/6)/\sin^2y_0=\kappa_5$ directly gives
Eq.~\eqref{eq:GRKdepth_small_q}.

\subsection{Proof of Theorem~\ref{thm:alpha_c}}

Normalize Eqs.~\eqref{eq:GRKdepth_expand} and
\eqref{eq:MED_boundary_layer} by $\kappa_0n2^{n/2}$ and retain the integer
minimum in the alternating depth. The one-to-one change of variables
\begin{equation}
\alpha+\frac mn=\frac{n-m}{nq_{\mathcal{A}}},
\qquad
\alpha+1=\frac{n-m}{n}\frac{1+q_{\mathcal{A}}}{q_{\mathcal{A}}}
\end{equation}
also gives $q_{\mathcal{G}}=q_{\mathcal{A}}/(1+q_{\mathcal{A}})$. Equating the two depths and cancelling the common factor $(n-m)/(nq_{\mathcal{A}})$ yields Eq.~\eqref{eq:qAc_integer} at $q_{\mathcal{A}}=q_{\mathcal{A},\mathrm{c}}$.

For uniqueness, let $A_j(\alpha)$ be the normalized alternating-depth
candidate associated with a fixed $j$. The sine inequality and the envelope
identity $g_*'(q_{\mathcal G})=v_{2,*}$ give
\begin{align}
A_j'(\alpha)
&=\frac{j}{\sqrt b\,|\sin(j\theta_2)|}\ge1,\\
\frac{d}{d\alpha}
\frac{\mathrm{E}_{\ell,\mathcal G}^{\min}
}{\kappa_0n2^{n/2}}
&=1-\frac{\sin(2v_{2,*})-v_{2,*}}{\kappa_2}
2^{-(n-m)/2}<1.
\label{eq:critical_slopes}
\end{align}
Thus the GRK depth minus every $A_j$ is strictly decreasing. Only finitely
many candidates are active on a compact $\alpha$ interval, so taking their
lower envelope preserves continuity and strict decrease. There can therefore
be at most one crossing. At and above $\alpha_{\mathrm{bd}}$, Case~II reduces
$S_{\mathcal A}$ to serial Grover search, whereas Eq.~\eqref{eq:GRKdepth_expand}
contains a strictly negative order-$\sqrt b$ correction; hence GRK is
shallower there. On $0<\alpha<\alpha_{\mathrm{bd}}$,
$q_{\mathrm{bd}}<q_{\mathcal A}<(n-m)/m$, so the crossing is the unique
physical root of Eq.~\eqref{eq:qAc_integer}. Inverting the change of variables
gives Eq.~\eqref{eq:alpha_c_from_q} and the ordering on its two sides.

\subsection{Proof of Corollary~\ref{cor:alpha_c_asymp}}

Consider the interior regime of the Case~I window
$b^{-1}\ll q_{\mathcal A,\mathrm c}^{\mathrm{asy}}\ll b^{1/2}$ specified in
Theorem~\ref{thm:MED_caseI}. In this regime, substituting
$\alpha+m/n=(n-m)/(nq_{\mathcal A,\mathrm c}^{\mathrm{asy}})$ and
$q_{\mathcal G}=q_{\mathcal A,\mathrm c}^{\mathrm{asy}}/
(1+q_{\mathcal A,\mathrm c}^{\mathrm{asy}})$ into
Eqs.~\eqref{eq:Enew_expand} and~\eqref{eq:GRKdepth_expand}, then cancelling
$(n-m)/n$, gives Eq.~\eqref{eq:qAc_implicit}. The depth difference is strictly
decreasing in $\alpha$, since its derivative consists of negative
$2^{-m/3}(\alpha+m/n)^{-2/3}$ and
$2^{-(n-m)/2}[\sin(2v_{2,*})-v_{2,*}]$ contributions. Thus the continuous
equation has at most one positive root in this regime. The large-$m$ solution
below verifies the interior ordering, while the Case~I--II boundary is covered
by the integer-aware Theorem~\ref{thm:alpha_c}.

Any positive root of Eq.~\eqref{eq:qAc_implicit} has
$q_{\mathcal A,\mathrm c}^{\mathrm{asy}}\to0$ when both $m$ and $n-m$
diverge. We may therefore use
\begin{equation*}
g_*(q_{\mathcal G})=g_*(0)+(\pi/6)q_{\mathcal G}
+\mathcal{O}(q_{\mathcal G}^2)
\end{equation*}
and define
\begin{equation*}
u_c=[(n-m)/(nq_{\mathcal A,\mathrm c}^{\mathrm{asy}})]^{1/3}
=[\alpha_{\mathrm c,\mathcal A}^{\mathrm{asy}}+m/n]^{1/3}.
\end{equation*}
Keeping the
leading terms in Eq.~\eqref{eq:qAc_implicit} gives the balance
\begin{multline}
\kappa_4 2^{-(n-m)/2}u_c^3
+\kappa_0\kappa_6\left(\frac{n-m}{n}\right)^{2/3}2^{-m/3}u_c\\
-\kappa_0\frac{n-m}{n}\simeq0.
\label{eq:cubic_balance}
\end{multline}
The left-hand side is strictly increasing for $u_c>0$, so this approximation
has one positive root. Its two dominant balances are
\begin{align}
n-3m\to\infty:\qquad
u_c&\simeq\kappa_6^{-1}\left(\frac{n-m}{n}\right)^{1/3}2^{m/3},\\
3m-n\to\infty:\qquad
u_c^3&\simeq\frac{\kappa_0}{\kappa_4}
\frac{n-m}{n}2^{(n-m)/2}.
\end{align}
Using $q_{\mathcal A,\mathrm c}^{\mathrm{asy}}=(n-m)/(nu_c^3)$ and then
Eq.~\eqref{eq:alpha_c_from_q} gives Eqs.~\eqref{eq:alpha_c_left} and
\eqref{eq:alpha_c_right}. In the first case,
$q_{\mathcal A,\mathrm c}^{\mathrm{asy}}b\to\kappa_6^3=9/8$, so the result is
only a boundary-scale continuation. In the second,
$q_{\mathcal A,\mathrm c}^{\mathrm{asy}}/b^{-1}\to\infty$ and
$q_{\mathcal A,\mathrm c}^{\mathrm{asy}}/b^{1/2}\to0$, which verifies the
interior ordering. Equation~\eqref{eq:cubic_balance} is the leading-order expansion of
Eq.~\eqref{eq:qAc_implicit}.

\subsection{Proof of Corollary~\ref{cor:alpha_c_max}}

In the large-$n$ limit,
$q_{\mathcal G}=(n-m)/[n(\alpha+1)]\to0$, so
Eq.~\eqref{eq:cubic_balance} controls the leading scaling. Since
$\alpha_{\mathrm c,\mathcal A}^{\mathrm{asy}}=u_c^3-m/n$ grows
exponentially, its stationarity condition is
$\mathrm{d}u_c/\mathrm{d}m=0$ to leading order. Solving
Eq.~\eqref{eq:cubic_balance} together with its $m$ derivative gives
\begin{subequations}
\begin{align}
&u_c\simeq\frac35\kappa_6^{-1}
\left(\frac{n-m}{n}\right)^{1/3}2^{m/3},\\
&2^{(3m-n)/2}\simeq\frac{50}{27}
\frac{\kappa_0\kappa_6^3}{\kappa_4}.
\label{eq:max_balance}
\end{align}
\end{subequations}
The second relation yields
$m_{\mathrm{opt}}\simeq n/3+1.2175$. Substitution into the first gives
\begin{multline}
\alpha_{\mathrm{c},\mathcal A}^{\max}
\simeq\left(\frac{2}{5}\right)^{5/3}
\kappa_0^{2/3}\kappa_4^{-2/3}\kappa_6^{-1}2^{n/3}
\approx0.29763\,2^{n/3}.
\end{multline}
The omitted subtraction $m_{\mathrm{opt}}/n=\mathcal{O}(1)$ is subleading.
Moreover, $q_{\mathcal A,\mathrm c}^{\mathrm{asy}}=(n-m)/(nu_c^3)$ and
$b=2^m$ give $q_{\mathcal A,\mathrm c}^{\mathrm{asy}}b=\Theta(1)$ at the
maximum. It therefore belongs to the formal boundary-scale continuation,
which completes the proof.

\section{Finite-size edge-operator corrections}
\label{app:boundary_proof}

The analytical results in Sec.~\ref{sec:analytical} optimize the depth of the alternating sequence $S_{\mathcal A}$ defined in Eq.~\eqref{eq:S_A}. It is a special case of the finite-size family $S_{\mathcal D}$ in Eq.~\eqref{eq:S_D}. Consider $T_{\tilde k}=G_nG_m^{\tilde k}$ and the edge mismatch $\delta=\bar k-\tilde k$. Then we have
\begin{equation}
S_{\mathcal{D}}(k,\bar k,\tilde k)
=T_{\tilde k}G_m^\delta T_{\tilde k}^{\,k}.
\label{eq:SD_factorized}
\end{equation}
The choice $\delta=0$ gives $S_{\mathcal A}(k+1,\tilde k)$, so Eq.~\eqref{eq:SD_factorized} isolates the mismatch in one local factor $\delta$.

The operator $G_m^\delta$ is a rotation by $2\delta\theta_2$ in the local two-dimensional subspace and acts trivially on $|\bar b\rangle$, where $\theta_2=\arcsin(1/\sqrt b)$. Its spectrum therefore gives the exact identity
\begin{equation}
\|G_m^\delta-I\|=2\bigl|\sin(\delta\theta_2)\bigr|.
\label{eq:norm_exact}
\end{equation}
Here $\|\cdot\|$ denotes the operator norm. Writing $a_{\mathcal D}=\langle\bar b|S_{\mathcal D}|s_n\rangle$ and $a_{\mathcal A}=\langle\bar b|T_{\tilde k}^{\,k+1}|s_n\rangle$, unitary invariance and the Cauchy--Schwarz inequality give $|a_{\mathcal D}-a_{\mathcal A}|\leq2|\sin(\delta\theta_2)|$. Since $\pr=1-|a|^2$, this yields
\begin{multline}
\bigl|\pr_{\mathcal{D}}(k,\bar k,\tilde k)
-\pr_{\mathcal{A}}(k+1,\tilde k)\bigr|\\
\leq4\bigl|\sin(\delta\theta_2)\bigr|
\bigl(1+\bigl|\sin(\delta\theta_2)\bigr|\bigr).
\label{eq:prob_bound}
\end{multline}
This bound holds for every $b$ and every integer $\delta$. For fixed $\delta$ and $b\gg1$, it is of order $\mathcal{O}(|\delta|/\sqrt b)$.

\begin{proposition}\label{prop:boundary_error}
Assume $b\gg1$, $K\gg1$, and an interior Case~I optimum satisfying $q_{\mathcal{A}}b\to\infty$ and $q_{\mathcal{A}}/\sqrt b\to0$. Let $|\delta|\ll\tilde k_*+1=\Theta\bigl((q_{\mathcal{A}}b)^{1/3}\bigr)$. The edge-corrected success probability satisfies
\begin{equation}
\pr_{\mathcal{D}}(k,\bar k,\tilde k)
=\pr_{\mathcal{A}}(k+1,\tilde k)
+\mathcal{O}\!\left(\frac{|\delta|}{\sqrt N}\right)
+\mathcal{O}\!\left(\frac{\delta^2}{b}\right).
\label{eq:PD_vs_PA}
\end{equation}
Here $\pr_{\mathcal{A}}(k+1,\tilde k)$ is the exact edge-free probability given by Lemma~\ref{lem:exactProb}.
\end{proposition}

\begin{proof}
For $b\gg1$, $G_m^\delta=I+2\delta\theta_2J_{xy}+\mathcal{O}((\delta\theta_2)^2)$, where $J_{xy}$ generates rotations in the local two-dimensional subspace. On the rotation subspace of $T_{\tilde k}$, the signed angle is $\Omega(\lambda)=2\arcsin(\sin\gamma\sin\lambda)$ with $\lambda=(\tilde k+1)\theta_2$. The edge mismatch changes it by
\begin{multline}
\Delta\Omega_\delta
=\Omega(\lambda+\delta\theta_2)-\Omega(\lambda)\\
=\frac{2\delta\theta_2\sin\gamma\cos\lambda}
{\sqrt{1-\sin^2\gamma\sin^2\lambda}}
+\mathcal{O}((\delta\theta_2)^2).
\label{eq:delta_Omega}
\end{multline}
At the interior optimum, $\lambda=(3q_{\mathcal A})^{1/3}b^{-1/6}(1+o(1))$, $\sin\gamma=K^{-1/2}$, and $\Omega\simeq2\sin\gamma\sin\lambda$. Hence the relative rotation-angle perturbation satisfies $|\Delta\Omega_\delta|/\Omega\simeq|\delta|/(\tilde k_*+1)=\mathcal{O}\bigl(|\delta|/(q_{\mathcal A}b)^{1/3}\bigr)$, which is small under the stated assumption. The corresponding change within the rotation subspace is $\mathcal{O}(|\delta|/\sqrt N)$. Equation~\eqref{eq:v_minus} also gives $\langle\bar b|v_-\rangle=\mathcal{O}(\sin\gamma)$, so the component injected into the $(-1)$ eigenspace has the same order, while the quadratic remainder is $\mathcal{O}(\delta^2/b)$. These estimates prove Eq.~\eqref{eq:PD_vs_PA}.
\end{proof}

Equation~\eqref{eq:prob_bound} is unconditional, whereas the proposition is local to a fixed small mismatch near an interior optimum. For $q_{\mathcal A}=\Theta(1)$, the proof gives a relative rotation-angle perturbation of order $\mathcal{O}(|\delta|/b^{1/3})$, and the single-run depth shift is exactly $\delta\ell_{G_m}$, which is subleading when $\delta=\mathcal{O}(1)$. Near the Case~I--II boundary, however, $\tilde k_*+1=\Theta(1)$, so no uniform $b^{-1/3}$ suppression follows. The edge degree of freedom may therefore shift finite-size transition points, and the unrestricted asymptotic optimization of $S_{\mathcal D}$ remains open.


%

\end{document}